%% file: main_arXiv.tex
\documentclass[journal]{IEEEtran}

\usepackage{amsfonts}
\usepackage{amsthm} % For theorem, lemma, definition, proof
\newtheorem{proposition}{Proposition}
\usepackage{algorithm}
\usepackage{algorithmic}
\usepackage{stfloats}
\usepackage{subcaption}  
\usepackage{multirow}
\usepackage[font=small,skip=6pt]{caption}
 
\usepackage{diagbox}
\usepackage{makecell}
\usepackage{amsmath,amssymb}
\usepackage{graphicx} 
\usepackage{booktabs} 
\usepackage{xcolor}
\usepackage{hyperref}
\usepackage{pgfplots}
\usepgfplotslibrary{groupplots}
\usetikzlibrary{arrows.meta,positioning}
\pgfplotsset{compat=1.18}
\definecolor{zkBlue}{RGB}{0,0,255}
\definecolor{zkRed}{RGB}{255,0,0}
\newcommand{\revblue}[1]{#1}
\newcommand{\revred}[1]{#1}
\pgfplotsset{
    zkamsAxis/.style={
        axis background/.style={fill=white},
        axis lines=left,
        axis line style={draw=black!70, -{Stealth[length=3pt,width=3pt]}},
        tick style={draw=black!70},
        ticklabel style={font=\small, color=black},
        label style={font=\small, color=black},
        title style={font=\small\bfseries, color=black},
        legend style={font=\small, draw=black!40, fill=white, text=black},
        grid=both,
        major grid style={draw=black!20, line width=0.4pt},
        minor grid style={draw=black!10, line width=0.2pt},
        minor tick num=1,
    },
    zkamsInsetAxis/.style={
        axis background/.style={fill=white},
        axis lines=left,
        axis line style={draw=black!70, -{Stealth[length=2.5pt,width=2.5pt]}},
        tick style={draw=black!70},
        ticklabel style={font=\scriptsize, color=black},
        label style={font=\scriptsize, color=black},
        grid=both,
        major grid style={draw=black!20, line width=0.4pt},
        minor grid style={draw=black!10, line width=0.2pt},
        minor tick num=1,
    },
    zkBlueLine/.style={color=zkBlue, very thick, mark size=2.3pt},
    zkRedLine/.style={color=zkRed, very thick, mark size=2.3pt},
}

\newcommand{\full}{\ensuremath{\bullet}} % 鈼?fully supported (math bullet, stable)
\newcommand{\none}{\ensuremath{\circ}}   % 鈼?not supported

\DeclareRobustCommand{\partSup}{%
\raisebox{0.03ex}{%
\begin{tikzpicture}[scale=0.055,baseline=-0.6ex]
	  \draw (0,0) circle (1);
	  \fill (0,0) -- (1,0) arc (0:180:1) -- cycle; % left half filled
\end{tikzpicture}}}

\usepackage{tabularx}
\usepackage{booktabs} % 鎺ㄨ崘鐢ㄤ簬鐢绘洿涓撲笟鐨勬í绾?
\usepackage{ragged2e} % 鎻愪緵鏇村ソ鐨勫乏瀵归綈鏁堟灉 \RaggedRight
\newcolumntype{L}[1]{>{\hsize=#1\hsize\RaggedRight\arraybackslash}X}
\usepackage{placeins} % 鐢ㄤ簬 \FloatBarrier
\usepackage[font=small,skip=0pt]{caption} % 璋冩暣Caption闂磋窛

\begin{document}

\title{ZK-AMS: Credibly Anonymous Admission for Web 3.0 Platforms via Recursive Proof Aggregation}

\author{Zibin~Lin,~Taotao~Wang,~Shengli~Zhang,~Long~Shi,~Boris~D{\"u}dder,~and~Shui~Yu   % <-this % stops a space

\thanks{Z. Lin, T. Wang and S. Zhang are with the College of Electronics and Information Engineering, Shenzhen University, Shenzhen 518052, China (e-mail: linaacc9595@gmail.com; ttwang@szu.edu.cn; zsl@szu.edu.cn)}% <-this % stops a space

\thanks{L. Shi is with the School of Electronic and Optical
Engineering, Nanjing University of Science and Technology, Nanjing 210094,
China (e-mail: longshi@njust.edu.cn)}% <-this % stops a space

\thanks{B. D{\"u}dder is with the Department of Computer Science, University of Copenhagen, Copenhagen, Denmark (e-mail: boris.d@di.ku.dk)}% <-this % stops a space

\thanks{Shui Yu is with the School of Computer Science, University of Technology
Sydney, Sydney, NSW 2007, Australia (e-mail: shui.yu@uts.edu.au)}% <-this % stops a space

\thanks{{\itshape Corresponding author: Taotao Wang (e-mail: ttwang@szu.edu.cn).}}
}

\markboth{ZK-AMS: Credibly Anonymous Admission for Web 3.0 Platforms via Recursive Proof Aggregation}%
{ZK-AMS: Credibly Anonymous Admission for Web 3.0 Platforms via Recursive Proof Aggregation}

\maketitle

\begin{abstract}
Web 3.0 platforms need an onboarding mechanism that can admit real users at scale without forcing them to reveal identity documents or pay one on-chain verification cost per user. Existing approaches typically rely on KYC-style disclosure, per-request on-chain verification, or trusted batching, making onboarding cost and latency difficult to predict under bursty demand. We present \textbf{ZK-AMS}, a credibly anonymous admission infrastructure that maps Personhood Credentials to anonymous on-chain Soul Accounts. Rather than introducing a new primitive, ZK-AMS composes zero-knowledge credential validation, permissionless batch submission, recursive proof aggregation, and anonymous post-admission account provisioning into one end-to-end workflow. Its key design feature is a confidential batching pipeline in which admission instances of a common relation are folded off-chain under multi-key homomorphic encryption, allowing an untrusted batch submitter to coordinate aggregation without direct access to individual user witnesses during batching; the confidentiality scope is characterized explicitly in the security analysis. The resulting batch is settled on-chain with constant verification cost per batch rather than per admitted user. We implement ZK-AMS on an Ethereum testbed and evaluate admission throughput, end-to-end latency, gas consumption, and parameter trade-offs. Results show stable batch-verification gas across evaluated batch sizes, substantially lower amortized on-chain cost than the non-recursive baseline, and practical cost-latency trade-offs for high-concurrency onboarding in Web 3.0 platforms.
\end{abstract}

\begin{IEEEkeywords}
privacy-preserving admission, decentralized identity, zero-knowledge proofs, recursive proof aggregation, multi-key homomorphic encryption, anonymous account provisioning.
\end{IEEEkeywords}

\IEEEpeerreviewmaketitle

\section{Introduction}

\IEEEPARstart{W}{eb 3.0} platforms are increasingly deployed as open, programmable computing environments, including DeFi systems, DAO governance, creator economies, and decentralized social applications~\cite{cao2025web}. In these platforms, user onboarding is not a peripheral feature but a platform-level primitive: the system must admit legitimate users, reject Sybil identities, and provision usable on-chain accounts under bursty and adversarial demand~\cite{wang2023account,liu2025detecting}. This requirement becomes particularly challenging when a platform seeks \emph{credibly anonymous} admission, namely, allowing users to remain anonymous at the account level while still giving the platform a verifiable basis for uniqueness and policy enforcement.

A practical Web 3.0 admission infrastructure should satisfy three requirements simultaneously: verifiable uniqueness, user anonymity, and predictable cost/latency under concurrent onboarding. Existing approaches typically satisfy only part of this design space. Soulbound-token-style identity signals improve accountability but do not provide protocol-level uniqueness guarantees~\cite{weyl2022decentralized}. KYC-style admission restores uniqueness but breaks privacy and introduces centralized bottlenecks~\cite{PiperArxiv2025Permissioning}. Zero-knowledge credential verification improves privacy, yet many designs still rely on per-request on-chain proof validation, causing admission cost to grow with the number of joiners~\cite{LauingerICBC2024Portal,RosenbergSP2023zkcreds,PiperArxiv2025Permissioning}. More broadly, batching often reduces settlement overhead only by assuming a trusted coordinator or exposing intermediate witness material, while identity-oriented schemes usually stop at credential verification and do not address anonymous post-admission account provisioning.

Beyond this high-level tension, Web 3.0 admission introduces three concrete obstacles. Prior work on decentralized computation, privacy protection, anonymous endorsement, and privacy-preserving blockchain processing suggests that decentralization, privacy, and public verifiability must often be engineered jointly rather than optimized in isolation~\cite{MendisTETC2021,RasheedTETC2022,MazumdarTETC2021}. In our setting, this means that admission requests must share a common verification relation for sound batch processing, a permissionless batch coordinator should aggregate requests without direct access to user witnesses during batching, and admission must be followed by anonymous but policy-compliant account provisioning. Together, these requirements make Web 3.0 admission a systems problem rather than a standalone credential-verification problem.

%Beyond this high-level tension, Web 3.0 admission introduces three concrete technical obstacles. First, admission requests must be expressed in a \emph{common verification relation} so that they can be batch-processed without sacrificing soundness. Second, if batching is delegated to a permissionless coordinator, the coordinator should be able to aggregate requests without learning user witnesses. Third, batch compression alone is insufficient for platform onboarding: after admission is settled, the system must still provision an operational on-chain account while enforcing one-to-one use of the admitted identity. These requirements together make Web 3.0 admission a systems problem rather than a standalone credential-verification problem.

This paper presents \textbf{ZK-AMS} (\textbf{Z}ero-\textbf{K}nowledge \textbf{A}dmission \textbf{M}apping \textbf{S}ystem), a platform-level admission infrastructure that maps Personhood Credentials~\cite{adler2024personhood} to anonymous on-chain accounts while preserving verifiable uniqueness. The core challenge is to achieve four properties simultaneously: batching that limits witness exposure against an untrusted Permissionless Batch Submitter (PBS) under explicit prototype assumptions, constant-cost on-chain settlement per batch, verifiable uniqueness rooted in real-world personhood credentials, and anonymous yet enforceable account provisioning. ZK-AMS addresses these challenges by combining zero-knowledge credential validation, Nova-style recursive proof aggregation~\cite{nova}, BGV-style multi-key homomorphic encryption~\cite{LopezAlt2012MultikeyFHEMPC,BrakerskiGentryVaikuntanathan2012}, and linkable ring signatures~\cite{noether2016ring} into a single end-to-end workflow.

%A preliminary version of part of this work appeared as zkBID in WWW'25~\cite{zkbid}. That paper focused on single-admission identity-to-account binding via PHC-based zero-knowledge verification and MLSAGS-based anonymous Soul-Account certification, together with prototype measurements of Groth16 and MLSAGS costs. The present journal version retains that binding mechanism but extends it to a scalable Web~3.0 admission infrastructure by adding common-relation recursive aggregation, permissionless encrypted off-chain batching under an untrusted PBS, batch-oriented proof finalization and constant-cost settlement, and expanded system/security analysis with end-to-end evaluation. The WWW'25 version did not include permissionless batching, encrypted witness folding, batch-level settlement, or the present end-to-end system analysis; thus, this manuscript should be read as a system-level extension rather than a simple experimental expansion.

A preliminary version of part of this work appeared as zkBID in WWW'25~\cite{zkbid}, focusing on single-admission identity-to-account binding via PHC-based Groth16 zero-knowledge verification and linkable ring signature based anonymous Soul-Account certification, together with prototype measurements of Groth16 and linkable ring signature  costs. The present journal version retains that binding mechanism but extends it to a scalable Web~3.0 admission infrastructure by adding common-relation recursive aggregation, permissionless encrypted off-chain batching under an untrusted PBS, batch-oriented proof finalization and constant-cost settlement, and expanded system/security analysis with end-to-end evaluation. Since the conference version did not include permissionless batching, encrypted witness folding, batch-level settlement, or the present end-to-end system analysis, this journal version should be understood as a system-level extension rather than a simple experimental expansion.

The main technical contributions of this paper are summarized as follows:
\begin{itemize}
	\item[1)] \textbf{A common-relation admission formulation for scalable Web 3.0 onboarding:}
	We formulate Web 3.0 admission as a committed relaxed rank-1 constraint system (R1CS) relation, namely a Nova-compatible relaxed R1CS admission relation whose witness-bearing vectors are represented to the batching layer through public commitments, thereby binding Personhood Credential validity, holder ownership, and public admission inputs into homogeneous admission instances suitable for recursive aggregation.
	
	\item[2)] \textbf{A recursive batching workflow with explicitly scoped confidentiality under an untrusted PBS:}
	We design a permissionless off-chain batching mechanism in which admission instances of the same relation are recursively folded over MKHE ciphertexts, while transcript-bound Fiat--Shamir challenges and public commitments preserve batch consistency without giving the batch coordinator direct access to individual user witnesses during batching; the precise confidentiality scope is stated in Section~\ref{securityAnalysis}.
	
	\item[3)] \textbf{A batch-settled admission and anonymous provisioning architecture:}
	We couple recursive batch settlement with a post-admission account-provisioning stage based on linkable ring signatures, so that the platform enforces one-to-one use of admitted identities while preserving account-level anonymity.
	
	\item[4)] \textbf{An implementation-driven evaluation of cost, latency, and deployment trade-offs:}
	We implement ZK-AMS on an Ethereum testbed and evaluate proof-generation cost, batch-settlement gas, end-to-end onboarding latency, and confidential-folding overhead, showing substantially improved amortized admission efficiency over the non-recursive baseline.
\end{itemize}

The remainder of this paper is organized as follows. Section~\ref{background} reviews the technical building blocks. Section~\ref{system_model} presents the system architecture and operational model. Section~\ref{detaildesign} details the ZK-AMS workflows. Section~\ref{securityAnalysis} analyzes security and operational risks. Section~\ref{experiment} presents the implementation and experimental evaluation. Section~\ref{rw} discusses related work, and Section~\ref{concl} concludes the paper.

% =================================================================================
% SECTION II: Technical Building Blocks
% =================================================================================
\section{Technical Building Blocks}\label{background}

ZK-AMS combines four standard components: PHCs as the off-chain uniqueness root, Nova-style recursive folding for batch compression, MKHE for witness-protected off-chain batching under the stated prototype assumptions, and MLSAGS for anonymous post-admission account provisioning. This section states only the architectural role of each component; formal definitions and protocol specifics are deferred to Appendices~\ref{supp:nova}, \ref{supp:mkhe}, and \ref{supp:lrs}.

\subsection{Personhood Credentials (PHCs)}
A Personhood Credential (PHC)~\cite{adler2024personhood} certifies that its holder corresponds to a unique human. In ZK-AMS, it provides the off-chain uniqueness root: the user proves in zero knowledge that the PHC is valid and correctly bound to the public admission input, without exposing identity attributes on-chain.

\subsection{Recursive Proof Aggregation via Nova Folding}
Nova-style folding~\cite{nova} compresses many admission instances of the same relation into one accumulated instance, so chain-side verification is paid once per batch rather than once per user. ZK-AMS uses this mechanism as the main scalability lever for high-concurrency onboarding; Appendix~\ref{supp:nova} gives the extended folding equations used in the system.

\subsection{Off-Chain Folding with Explicitly Scoped Confidentiality via Multi-Key Homomorphic Encryption}
Recursive aggregation alone is insufficient if the batch coordinator can inspect witness data. ZK-AMS therefore performs the folding-relevant computations on BGV-style MKHE ciphertexts~\cite{LopezAlt2012MultikeyFHEMPC,BrakerskiGentryVaikuntanathan2012}, so the PBS coordinates batching over ciphertexts and public commitments rather than individual plaintext witnesses. Only the randomized folded batch state needed for proof finalization is collaboratively revealed at batch end, and the exact confidentiality scope is stated in Section~\ref{securityAnalysis} and Appendix~\ref{supp:mkhe}.

\subsection{Anonymous Account Provisioning via Linkable Ring Signatures}
After batch admission, an admitted user still needs an operational account. ZK-AMS uses MLSAGS~\cite{noether2016ring}, together with the broader line of ring-signature techniques~\cite{rivest2001leak,beullens2020calamari}, to let a user bind a fresh Soul Account to the admitted seed identity anonymously while a deterministic key image prevents repeated provisioning attempts. Appendix~\ref{supp:lrs} summarizes the MLSAGS interface and the key-image mechanism used in this provisioning step.

\subsection{Architectural Role of the Building Blocks}
In ZK-AMS, PHCs provide the off-chain uniqueness root, Nova compresses homogeneous admission batches, MKHE limits witness exposure during batching under the stated prototype assumptions, and MLSAGS enables anonymous but enforceable post-admission provisioning. The contribution of ZK-AMS is therefore the system integration of these components into a Web~3.0 admission workflow with predictable batch settlement, rather than a new cryptographic primitive.

% =================================================================================
% SECTION III: System Architecture and Operational Model
% =================================================================================
\section{System Architecture and Operational Model}\label{system_model}

ZK-AMS (Zero-Knowledge Admission Mapping System) is a platform-level onboarding infrastructure for Web 3.0 environments that maps a real-world uniqueness credential to a credibly anonymous on-chain account without per-user on-chain verification. The provisioned account is termed a \emph{Soul Account}. Fig.~\ref{fig:zkams_overview} summarizes the end-to-end architecture, spanning a client-side credential layer, an off-chain batching layer with explicitly scoped confidentiality, and an on-chain settlement and provisioning layer.

\begin{figure*}[!t]
    \centering
    \includegraphics[width=0.85\textwidth]{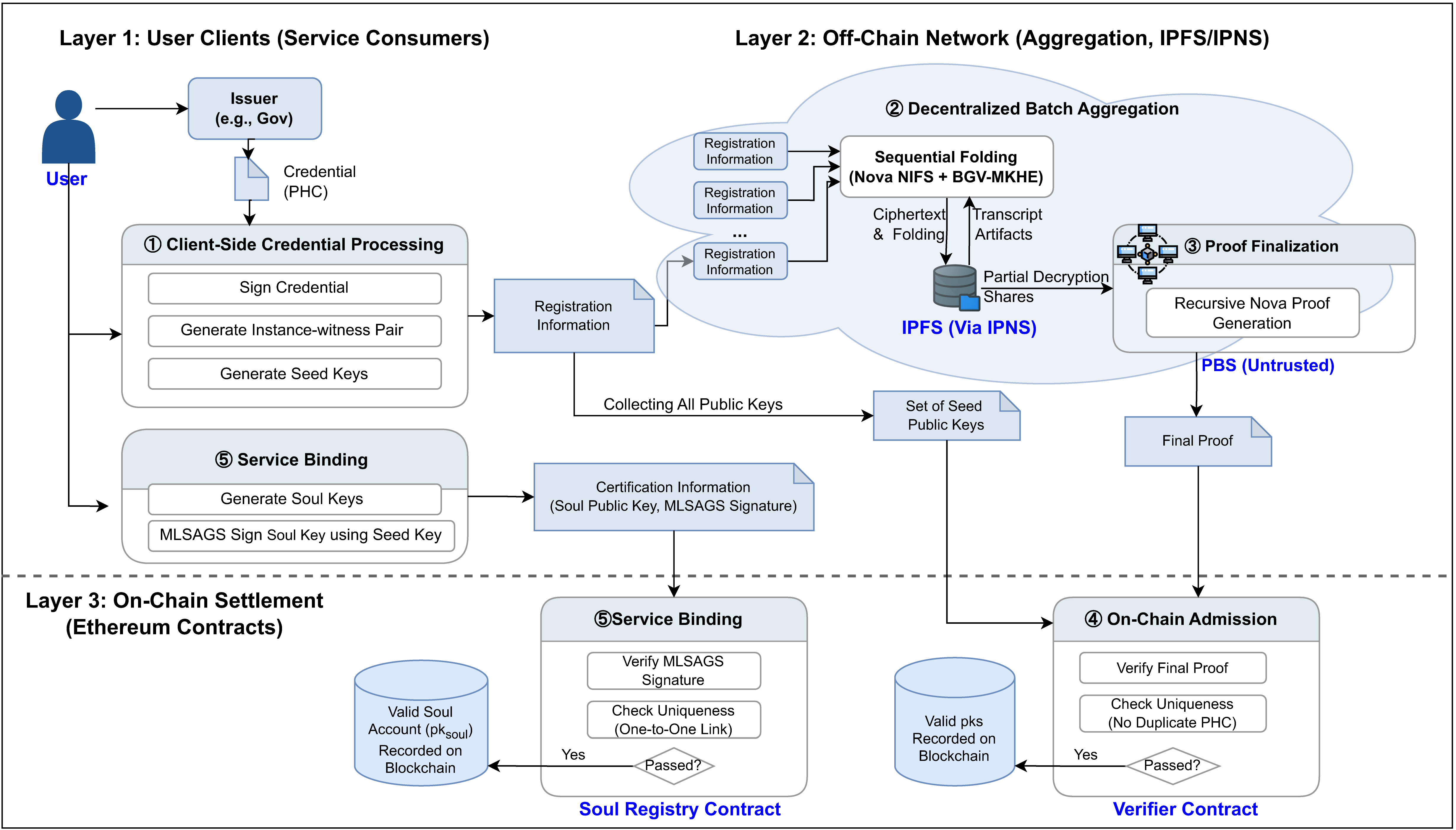}
    \caption{ZK-AMS platform architecture and admission workflow:
    (1) client-side credential processing;
    (2) confidential off-chain batch aggregation over IPFS/IPNS;
    (3) PBS proof finalization;
    (4) on-chain batch admission; and
    (5) anonymous account provisioning via the Soul Registry.}
    \label{fig:zkams_overview}
    \vspace{-0.4cm}
\end{figure*}

\subsection{Architectural Roles and Data Planes}
ZK-AMS separates an off-chain data/compute plane from an on-chain settlement plane. The main roles are:
\begin{itemize}
    \item \textbf{Users:}
    Each user holds a PHC, transforms it into a privacy-preserving admission credential, participates in encrypted off-chain batching, and later provisions a Soul Account after successful admission.

    \item \textbf{Permissionless Batch Submitter (PBS):}
    Any node may act as the PBS by collecting the public off-chain transcript, obtaining the folded admission material, generating the succinct proof, and submitting the batch to the blockchain. The PBS is not trusted with individual user witnesses.

    \item \textbf{On-chain contracts:}
    A Verifier Contract validates the batch proof and records admission anchors, such as PHC hashes and seed public keys, to prevent duplicate registration. A Soul Registry Contract verifies anonymous account binding and enforces one-to-one provisioning through key-image uniqueness.

    \item \textbf{Off-chain transcript layer:}
    Intermediate batching artifacts, including encrypted folding state, public commitments, and decryption shares, are stored on IPFS~\cite{ipfs_docs} and discovered via IPNS~\cite{ipfs_ipns_docs}. Only compact settlement anchors are posted on-chain.

\end{itemize}

This separation keeps cryptographic and batching complexity off-chain while reserving public verifiability and policy enforcement for the chain. As a result, ZK-AMS treats admission as a scalable computing workflow rather than a per-user smart-contract verification task.

\subsection{Admission Relation and Batch Semantics}
All admissions in a batch correspond to the \emph{same} admission relation. Let $\mathcal{R}_{adm}$ denote the committed relaxed R1CS relation capturing PHC validity, holder ownership, and public-input binding. Each user locally transforms a registration request into an instance--witness pair $(\mathcal{I}_i,\mathcal{W}_i)\in \mathcal{R}_{adm}$, where the public instance contains commitments and public admission inputs, and the witness contains the corresponding openings and private credential material.

The batching process is therefore \emph{homogeneous}: it aggregates many instances of the same admission relation rather than heterogeneous statements. This is required because recursive folding preserves consistency only when all folded instances belong to a common relation. Operationally, a batch is a set of admission requests sharing the same circuit, public system parameters, and verification semantics.

\subsection{Admission Lifecycle}
ZK-AMS executes the following five-phase admission lifecycle, detailed in Section~\ref{detaildesign}.

\begin{enumerate}
    \item \textbf{Client-side credential processing:}
    The user converts a PHC into an admission credential in committed instance--witness form and generates a seed key pair.

    \item \textbf{Confidential off-chain batch aggregation:}
    Users upload encrypted folding artifacts to the off-chain transcript layer and collaboratively update the batch accumulator without revealing individual witnesses.

    \item \textbf{Proof finalization:}
    The PBS reconstructs the folded batch material required for final proof generation and produces a single succinct batch proof.

    \item \textbf{On-chain batch admission:}
    The PBS submits the batch proof and compressed batch input list to the Verifier Contract, which settles the batch and records the admission anchors.

    \item \textbf{Anonymous account provisioning:}
    Each admitted user binds a fresh Soul Account to the admitted seed identity through a linkable ring signature, and the Soul Registry Contract enforces one-to-one provisioning.
\end{enumerate}

This lifecycle turns onboarding into a reusable batch workflow: uniqueness is checked once, settlement is amortized, and account provisioning remains anonymous but enforceable.

\subsection{Cost- and Capacity-Predictable Admission Model}
From a computing-systems perspective, the main objective of ZK-AMS is not only privacy and uniqueness, but also \emph{operator-visible predictability} in settlement cost and onboarding latency. In non-recursive designs, admitting $N$ users requires verifying $N$ independent proofs on-chain, so the total settlement cost grows linearly:
\begin{equation}\label{equ:usercost1}
    C_{total}(N) = N \times (C_{verify} + C_{store}),
\end{equation}
where $C_{verify}$ is the gas cost of verifying one proof and $C_{store}$ captures per-user state updates. This linear behavior makes bursty onboarding events difficult to budget and provision.

ZK-AMS instead settles $N$ admissions using a single batch proof. The amortized chain-side cost per admitted user becomes
\begin{equation}\label{equ:usercost2}
    C_{user}(N) = \frac{C_{verifyFolded} + C_{overhead}}{N} + C_{stateUpdate}.
\end{equation}
where $C_{verifyFolded}$ is the cost of verifying one folded batch proof, $C_{overhead}$ captures batch-level submission overhead, and $C_{stateUpdate}$ is the per-user admission-state update cost. As $N$ increases, the fixed verification cost is shared across more users, and the amortized settlement cost approaches the minimum cost of a simple state update.

Admission latency exhibits the opposite trade-off. A larger batch reduces amortized on-chain cost, but also increases waiting, aggregation, and proving delay. At a high level, the end-to-end batch latency can be decomposed as
\begin{equation}\label{equ:batch_latency}
    T_{batch}(N) = T_{wait}(N) + T_{agg}(N) + T_{prove}(N) + T_{chain},
\end{equation}
where $T_{wait}(N)$ is the time needed to accumulate a batch, $T_{agg}(N)$ is the off-chain aggregation time, $T_{prove}(N)$ is the PBS proof-finalization time, and $T_{chain}$ is the on-chain settlement delay. Equations~\eqref{equ:usercost2} and \eqref{equ:batch_latency} together provide a practical configuration interface for platform operators: larger batches reduce per-user settlement cost, while smaller batches reduce onboarding delay.

This operational model distinguishes ZK-AMS from per-request admission schemes. Instead of forcing operators to pay and provision linearly with the number of joiners, ZK-AMS exposes an explicit batch-size knob for balancing throughput, latency, and cost under high-concurrency onboarding.

% =================================================================================
% SECTION IV: ZK-AMS Admission and Provisioning Workflows
% =================================================================================
\section{ZK-AMS Admission and Provisioning Workflows}\label{detaildesign}

This section specifies the end-to-end admission and provisioning workflow of ZK-AMS. Following Section~\ref{system_model}, the workflow consists of five phases: client-side credential processing, confidential off-chain batch aggregation, proof finalization, on-chain batch admission, and anonymous account provisioning. Together, these phases map a real-world uniqueness credential to an operational Soul Account.

All batched requests correspond to the same admission relation $\mathcal{R}_{adm}$. Hence, the off-chain aggregation phase folds a homogeneous sequence of committed relaxed R1CS instances rather than heterogeneous statements, which is essential to the soundness and operational simplicity of batch settlement.

% ---------------------------------------------------------------------------------
% Phase I
% ---------------------------------------------------------------------------------
\subsection{Phase I: Client-Side Credential Processing}\label{sec:register_info_generation}

This phase converts a Personhood Credential (PHC) into a machine-verifiable admission instance while keeping privacy-sensitive processing on the user side. The PHC contains the holder public key, the issuer public key, and the issuer signature, denoted by $\mathsf{pk_{holder}}$, $\mathsf{pk_{issuer}}$, and $\mathsf{sig_{issuer}}$, respectively.

The user forms a witness--public-input pair for the admission relation, where $W=\mathsf{PHC}$ and $x=\{\mathsf{hash_{PHC}}, \mathsf{sig_{holder}}, \mathsf{pk_{issuer}}\}$. Here, $\mathsf{hash_{PHC}}$ is the credential hash, and $\mathsf{sig_{holder}}$ is the holder signature on $\mathsf{hash_{PHC}}$ under the secret key corresponding to $\mathsf{pk_{holder}}$. The admission relation enforces the following three conditions:
\begin{equation}
\begin{aligned}
    \text{(i)} ~  & x.\mathsf{hash_{PHC}} = \mathsf{Hash}(\mathsf{Canon}(W)) \\
    \text{(ii)} ~ & 1 = \mathsf{SigVerify}(x.\mathsf{pk_{issuer}}, x.\mathsf{hash_{PHC}}, W.\mathsf{sig_{issuer}}) \\
    \text{(iii)}~ & 1 = \mathsf{SigVerify}(W.\mathsf{pk_{holder}}, x.\mathsf{hash_{PHC}}, x.\mathsf{sig_{holder}})
\end{aligned}
\end{equation}
where $\mathsf{Canon}(\cdot)$ denotes a deterministic canonicalization procedure and $\mathsf{SigVerify}(\cdot)$ is the signature-verification function. These checks respectively enforce credential integrity, issuer-certified validity, and holder ownership.

To make the admission request compatible with recursive aggregation, the user transforms $(x,W)$ into a committed relaxed R1CS instance--witness pair $(\mathcal{I},\mathcal{W})$ of the common relation $\mathcal{R}_{adm}$. Concretely, the relaxed R1CS vectors that would normally remain explicit are represented to the batching layer through public commitments, while their openings remain inside the witness. ZK-AMS uses an additively homomorphic linear commitment~\cite{Baum18}, compatible with BGV-style evaluation over the plaintext ring $R_t$. Let $\mathsf{pp_{commit}}=(G,H)$, where $G\in R_t^{k\times \ell}$ and $H\in R_t^{k\times m}$ are public matrices over $R_t$. For a message vector $M\in R_t^{m}$ and randomness vector $r\in R_t^{\ell}$, the commitment is defined as\footnotemark
\begin{equation}\label{equ:new_commit}
 \mathsf{Com}(\mathsf{pp_{commit}}, M, r) := G\cdot r + H\cdot M \in R_t^{k}.
\end{equation}
The matrices $G$ and $H$ are generated once from public randomness and thereafter treated as system parameters.
\footnotetext{Since messages from Nova/R1CS are over $\mathbb{F}_p$ but BGV and this commitment compute over $R_t$, we choose $t$ with $p\mid t$ (e.g., $t=p$) and embed $\mathbb{F}_p$ into $R_t$ by the natural canonical lift. This technique of Field-to-Ring Embedding ensures homomorphic plaintext operations in $R_t$ implement the same arithmetic as in $\mathbb{F}_p$~\cite{SV14}.}  

\revblue{This commitment instantiation does not introduce an additional trusted setup beyond public parameter generation. Under standard lattice assumptions over $R_t$ (e.g., Module-SIS/LWE-type assumptions), it provides computational binding, while hiding follows from fresh uniform randomness $r$. Its additive homomorphism preserves linear combinations, which is exactly the algebraic property needed to evaluate commitments and folding updates through BGV homomorphic operations on ciphertexts.}

\revblue{After sampling commitment randomness $(r_E,r_W)$, the user sets $E\gets\mathbf{0}$ and $u\gets1$ to embed a standard R1CS instance into the relaxed form, computes $\overline{E}$ and $\overline{W}$, and thereby obtains the committed relaxed R1CS pair $\mathcal{I}=(\overline{E},u,\overline{W},x)$ and $\mathcal{W}=(E,r_E,W,r_W)$.}

The resulting committed public instance is $\mathcal{I} = (\overline{E}, u, \overline{W}, x)$, and the corresponding witness is $\mathcal{W} = (E, r_E, W, r_W)$. Here and throughout the remainder of the paper, overlined symbols denote public commitments, while the corresponding unbarred symbols denote the underlying plaintext quantities; for example, $\overline{E}=\mathsf{Com}(\mathsf{pp}^{E}_{\mathsf{commit}}, E, r_E)$. The user then generates a seed key pair $(\mathsf{pk_{seed}}, \mathsf{sk_{seed}})$ for later account provisioning and forms the registration credential $\mathsf{resCredential}=(\mathsf{hash_{PHC}}, \mathsf{pk_{seed}}, \mathcal{I}, \mathcal{W})$. Algorithm~\ref{alg:R1CSinstance} summarizes the procedure.

\begin{algorithm}[!t]
\caption{Client-Side Instance Generation}
\label{alg:R1CSinstance}
\begin{algorithmic}[1]
\STATE \textbf{Input:} \ensuremath{W=\{\mathsf{PHC}\}}, \ensuremath{x=\{\mathsf{hash_{PHC}}, \mathsf{sig_{holder}},\mathsf{pk_{issuer}}\}}
\STATE \ensuremath{r_E\gets \mathsf{random}()}
\STATE \ensuremath{r_W \gets \mathsf{random}()}
\STATE \ensuremath{E \gets \mathbf{0}}
\STATE \ensuremath{u \gets 1}
\STATE \ensuremath{\overline{E} \gets \mathsf{Com}(\mathsf{pp}^{E}_{\mathsf{commit}},E,r_E)}
\STATE \ensuremath{\overline{W} \gets \mathsf{Com}(\mathsf{pp}^{W}_{\mathsf{commit}},W,r_W)}
\STATE \ensuremath{\mathcal{I} \gets (\overline{E}, u, \overline{W}, x)}
\STATE \ensuremath{\mathcal{W} \gets (E, r_E, W, r_W)}
\STATE \textbf{Return:} \ensuremath{(\mathcal{I},\mathcal{W})}
\end{algorithmic}
\end{algorithm}

% ---------------------------------------------------------------------------------
% Phase II
% ---------------------------------------------------------------------------------
\subsection{Phase II: Decentralized Batch Aggregation}\label{sec:decentralized_folding}

This phase implements the off-chain batch aggregation layer of ZK-AMS. A batch of $N$ admission instances of the same relation $\mathcal{R}_{adm}$ is folded into one accumulated statement and later settled with constant on-chain verification cost per batch. The folding pipeline runs over MKHE ciphertexts so that participating users update the encrypted accumulator state, while the PBS reads the public transcript and performs only batch finalization after folding completes. The workflow proceeds in three stages:

\subsubsection{Stage 1 (Initialization and Encryption)}
The system adopts globally agreed BGV MKHE parameters. Each user generates a key pair $(\mathsf{pk}_i, \mathsf{sk}_i)$ and publishes the public key through the off-chain transcript layer. The user then encrypts the folding-relevant variables in its local committed relaxed R1CS pair under $\mathsf{pk}_i$. Specifically, the encrypted plaintext variables are $g\in\{x_i,u_i,E_i,W_i,r_{E_i},r_{W_i}\}$, yielding ciphertexts $\{{c}_{g}\}_{g \in \{x_i,u_i,E_i,W_i,r_{E_i},r_{W_i}\}}$. Together with the public commitments $\overline{E}_i$ and $\overline{W}_i$, these encrypted artifacts are uploaded to the off-chain transcript layer.

\subsubsection{Stage 2 (Homomorphic State Transition)}
The core confidential batching component of ZK-AMS is the encrypted folding layer, which updates the batch accumulator state by applying Nova-style folding directly over ciphertexts. Appendices~\ref{supp:nova} and \ref{supp:mkhe} summarize the Nova and MKHE mechanics used in this stage.

 \textbf{Initial user ($h_1$):} \revblue{The first participant initializes the accumulator by posting the encrypted state and public commitments. Concretely, for each $g\in\{x_1,u_1,E_1,W_1,r_{E_1},r_{W_1}\}$, the accumulator ciphertext is initialized as $\hat{c}_{g_{acc}} \gets {c}_{g}$, and the public commitments are initialized as $(\overline{E}_{acc},\overline{W}_{acc}) \gets (\overline{E}_1,\overline{W}_1)$.}

 \textbf{Subsequent user ($h_i$):} A later participant retrieves the current accumulator from the public transcript and performs one homomorphic folding step:
    \begin{enumerate}
        \item[a)] \textbf{Cross-term processing:} Let $Z=(W,x,u)$ denote the witness/public-input vector used in the relaxed R1CS formulation, and denote the encrypted accumulated and incoming vectors by $\hat{c}_{Z_{acc}}=(\hat{c}_{W_{acc}},\hat{c}_{x_{acc}},\hat{c}_{u_{acc}})$ and ${c}_{Z_i}=({c}_{W_i},{c}_{x_i},{c}_{u_i})$. For compactness, we denote by $\hat{c}_{AZ_{acc}}$, $\hat{c}_{BZ_{acc}}$, $\hat{c}_{CZ_{acc}}$, $\hat{c}_{AZ_i}$, $\hat{c}_{BZ_i}$, and $\hat{c}_{CZ_i}$ the ciphertexts obtained by homomorphically evaluating the corresponding linear maps under the common matrices $(A,B,C)$. In the following equations, $\oplus$ denotes ciphertext-domain homomorphic addition, $\otimes$ denotes ciphertext-domain homomorphic multiplication, and $a\odot c$ denotes multiplication of ciphertext $c$ by a public scalar $a$ in the plaintext ring. The user first computes these linear images and then evaluates the Nova-style cross-term ciphertext:
\begin{equation}
\begin{aligned}
\hat{c}_T
&=
\hat{c}_{AZ_{acc}} \otimes \hat{c}_{BZ_i}
\oplus
\hat{c}_{AZ_i} \otimes \hat{c}_{BZ_{acc}}
\oplus  \\
&
(-1)\odot \hat{c}_{u_{acc}} \otimes\hat{c}_{CZ_i}
\oplus
(-1)\odot {c}_{u_i} \otimes \hat{c}_{CZ_{acc}}.
\end{aligned}
\end{equation}

        \item[b)] \textbf{Homomorphic commitment to the cross-term:} To avoid revealing the plaintext cross-term before batch materialization, the participant computes a public commitment to it directly in the encrypted domain. After sampling commitment randomness $r_T$ and encrypting it as ${c}_{r_T}$, the encrypted commitment is evaluated as
\begin{equation}
\hat{c}_{\overline{T}}
=
\mathsf{Com}(\mathsf{pp}^{T}_{\mathsf{commit}}, \hat{c}_T, {c}_{r_T})
=
G_T\cdot {c}_{r_T} \oplus H_T\cdot \hat{c}_T.
\end{equation}
        Users collaboratively decrypt only $\hat{c}_{\overline{T}}$ to obtain the public commitment value $\overline{T}$, which becomes part of the transcript. The plaintext cross-term itself remains hidden.

        \item[c)] \textbf{Fiat--Shamir challenge generation:} \revblue{The folding challenge is derived from the public transcript via the Fiat--Shamir transform. Let $\mathsf{ctx}$ denote a domain-separation string including the protocol identifier, the current transcript or batch identifier, and a digest of $(\mathsf{vk_{NIFS}}, X)$, where $\mathsf{vk_{NIFS}}$ is the public verification key for the Nova-style folding interface and $X = \{ {\mathsf{hash_{PHC}}}_i, {\mathsf{pk_{seed}}}_i \}_{i=1}^N$ is the compressed batch input list later submitted at settlement.} A Fiat--Shamir challenge is then derived from the public transcript:
\begin{equation}\label{equ:FS-challenge}
\begin{aligned}
v := &  H(\mathsf{ctx},
       \overline{E}_{acc}, \overline{W}_{acc},
       \overline{E}_i,    \overline{W}_i,
       \overline{T}, i).
\end{aligned}
\end{equation}
        \revblue{This challenge binds the folding step and the settled batch metadata to the same public transcript. Hence, any party can recompute $v$, and any post hoc rewrite of the transcript or the ordered batch input list $X$ invalidates the final proof.}

        \item[d)] \textbf{Homomorphic folding of encrypted variables:} Using $v$, the participant updates the encrypted accumulated state:
\begin{equation}
\begin{aligned}
{\hat{c}_{x_{new}} = \hat{c}_{x_{acc}} \oplus (v \odot {c}_{x_i})}\\
{\hat{c}_{u_{new}} = \hat{c}_{u_{acc}} \oplus (v \odot {c}_{u_i})}
\end{aligned}
\end{equation}
and
\begin{equation}
\begin{aligned}
\hat{c}_{E_{new}} &= \hat{c}_{E_{acc}} \oplus (v \odot \hat{c}_T) \oplus (v^2 \odot {c}_{E_i}), \\
\hat{c}_{r_{E_{new}}} &= \hat{c}_{r_{E_{acc}}} \oplus (v \odot c_{r_T}) \oplus (v^2 \odot {c}_{r_{E_i}}), \\
\hat{c}_{W_{new}} &= \hat{c}_{W_{acc}} \oplus (v \odot {c}_{W_i}), \\
\hat{c}_{r_{W_{new}}} &= \hat{c}_{r_{W_{acc}}} \oplus (v \odot {c}_{r_{W_i}}).
\end{aligned}
\end{equation}

        \item[e)] \textbf{Public commitment folding:} The corresponding public commitments are folded consistently:
\begin{equation}
\begin{aligned}
\overline{E}_{new} &= \overline{E}_{acc} + v \cdot \overline{T} + v^2 \cdot \overline{E}_i, \\
\overline{W}_{new} &= \overline{W}_{acc} + v \cdot \overline{W}_i.
\end{aligned}
\end{equation}
        
        \item[f)] \textbf{State update:} \revblue{The new encrypted state and commitments are then written back to the transcript. For each $g\in\{x,u,E,r_E,W,r_W\}$, set $\hat{c}_{g_{acc}} \gets \hat{c}_{g_{new}}$, and update the public commitments as $(\overline{E}_{acc},\overline{W}_{acc})\gets(\overline{E}_{new},\overline{W}_{new})$.}
    \end{enumerate}

\subsubsection{Stage 3 (Distributed Decryption and Batch Materialization)}
After the last folding step, users collaboratively reconstruct only the folded plaintext components required for proof finalization. Appendix~\ref{supp:mkhe} gives the MKHE decryption interfaces and the prototype-specific share-release assumptions used here.
\begin{itemize}
    \item \textbf{Partial decryption:} Let $\hat{c}_{acc}= \{\hat{c}_{x}, \hat{c}_{E}, \hat{c}_{r_E}, \hat{c}_{W}, \hat{c}_{r_W}, \hat{c}_{u}\}$ denote the final accumulator ciphertexts. Each user computes partial decryption shares $\mu_{i, g} \leftarrow \mathsf{MKHE.PartialDecrypt}(\mathsf{sk}_i, \hat{c}_{g})$ for each relevant component $g$ and publishes them to the transcript layer.
    \item \textbf{Share fusion:} The PBS or any other eligible node combines the published shares to reconstruct $x_f,u_f,E_f,r_{E,f},W_f,r_{W,f}$. In the current prototype, the evaluated accumulator depends on the full batch key set, so this step uses the full set of published shares for the finalized batch rather than a separate tunable $t$-out-of-$N$ threshold-decryption policy.
    \item \textbf{Batch materialization:} The accumulated public instance becomes $\mathcal{I}_{acc,N}=(\overline{E}_{acc},u_f,\overline{W}_{acc},x_f)$, which serves as the public batch anchor for proof finalization. Because the folded components are randomized linear combinations, they do not reveal individual user witnesses.
\end{itemize}

% ---------------------------------------------------------------------------------
% Phase III
% ---------------------------------------------------------------------------------
\subsection{Phase III: Proof Finalization and Submission}\label{sec:recursive_proof_generation}

This phase converts the materialized folded batch into a single succinct settlement proof and prepares the corresponding batch-submission material. Once a batch reaches the configured size $N$ or a timeout triggers finalization, the PBS completes proof finalization against the accumulated admission relation. This phase proceeds through the following stages.

\subsubsection{Finalization Trigger and Objective}
Once a batch reaches size $N$ or a timeout triggers finalization, the PBS generates a succinct proof that the accumulated batch is folding-consistent and satisfiable under the admission relation. \revblue{The objective of this phase is to certify both that the final public folding step is valid and that the resulting committed relaxed R1CS instance still satisfies the same admission relation enforced in Phase~I.}

\subsubsection{Accumulator Reconstruction}
After distributed decryption, the PBS receives $(x_f,u_f,E_f,r_{E,f},W_f,r_{W,f})$, forms the accumulated witness $\mathcal{W}_{acc,N}=(E_f,r_{E,f},W_f,r_{W,f})$, and reconstructs the public instance $\mathcal{I}_{acc,N}=(\overline{E}_{acc},u_f,\overline{W}_{acc},x_f)$ from the transcript commitments. \revblue{In addition, it reads the public transcript metadata needed to recompute the Fiat--Shamir challenge for the finalization step.}

\subsubsection{Padding Fold and Proof Statement}
To finalize the recursive chain, ZK-AMS performs one publicly defined padding fold using a fixed padding instance--witness pair $(\mathcal{I}_{pad},\mathcal{W}_{pad})$. This step preserves the admission relation while producing a terminal folded instance suitable for succinct proof generation. \revblue{Since the PBS already holds the folded plaintext components needed for the last cross-term evaluation, it computes the padding cross-term $T_{N+1}$ in the clear using $(x_f,u_f,W_f)$ and the fixed padding values $(x_{pad},u_{pad},W_{pad})$, and then commits to it to obtain $\overline{T}_{N+1}$.}

\revblue{We denote by $\mathcal{C}_1$ the committed relaxed R1CS checker for the admission relation. On input $(\mathcal{I}=(\overline{E},u,\overline{W},x),\mathcal{W}=(E,r_E,W,r_W))$, $\mathcal{C}_1$ checks (i) commitment-opening consistency, namely $\overline{E}=\mathsf{Com}(\mathsf{pp}^{E}_{\mathsf{commit}},E,r_E)$ and $\overline{W}=\mathsf{Com}(\mathsf{pp}^{W}_{\mathsf{commit}},W,r_W)$, and (ii) relaxed R1CS satisfiability $(AZ)\circ(BZ)=u\cdot(CZ)+E$ for $Z=(W,x,u)$, where the encoded constraints are exactly the PHC integrity, issuer-validity, and holder-ownership checks from Phase~I.}

\revblue{We then design a finalization circuit $\mathcal{C}_2$ that proves the materialized batch is both transcript-consistent and satisfiable. Specifically, $\mathcal{C}_2$ (i) recomputes $\mathcal{I}_{acc,N+1}=\mathsf{NIFS.Verify}(\mathsf{vk_{NIFS}},\mathcal{I}_{acc,N},\mathcal{I}_{pad},\overline{T}_{N+1})$ to enforce consistency of the final padding fold, (ii) runs $\mathcal{C}_1$ on $(\mathcal{I}_{acc,N+1},\mathcal{W}_{acc,N+1})$, and (iii) enforces that $\mathcal{I}_{pad}$ equals the publicly specified padding instance so that the PBS cannot substitute an arbitrary padding input. The resulting proof therefore certifies that the submitted batch anchor corresponds to a well-formed finalized fold and a satisfiable accumulated admission statement. Appendix~\ref{supp:nova} summarizes the Nova-side background for this finalization step.}

\subsubsection{Proof-Generation Interface}
For the fixed proof circuit $\mathcal{C}_2$, the proving and verification keys are generated once during system setup and thereafter treated as public system parameters for the PBS and the on-chain verifier. Algorithm~\ref{alg:zkpgeneration} summarizes proof finalization.

\begin{algorithm}[!t]
\caption{Recursive Proof Finalization}
\label{alg:zkpgeneration}
\begin{algorithmic}[1]
\STATE \textbf{Input:} \ensuremath{\mathcal{I}_{acc,N}, \overline{T}_{N+1}, \mathcal{I}_{pad}, \mathcal{I}_{acc,N+1}, \mathcal{W}_{acc,N+1}, \mathsf{pk}}
\STATE \textbf{Public parameters:} verification key \ensuremath{\mathsf{vk}} for circuit \ensuremath{\mathcal{C}_2}, generated once during system setup
\STATE \textbf{Generate proof:}
\STATE \quad \ensuremath{\mathsf{Input_{pub}} \gets (\mathcal{I}_{acc,N},\overline{T}_{N+1},\mathcal{I}_{acc,N+1})}
\STATE \quad \ensuremath{\mathsf{Input_{pri}} \gets (\mathcal{I}_{pad},\mathcal{W}_{acc,N+1})}
\STATE \quad \ensuremath{\pi_{zk} \gets \mathsf{zkSNARK.PROVE}(\mathsf{pk}, \mathsf{Input_{pub}}, \mathsf{Input_{pri}})}
\STATE \textbf{Return:} \ensuremath{\pi_{zk}}
\end{algorithmic}
\end{algorithm}

Operationally, this phase bridges the off-chain batching layer and the on-chain settlement layer without giving the PBS direct access to individual user witnesses during batching.

% ---------------------------------------------------------------------------------
% Phase IV
% ---------------------------------------------------------------------------------
\subsection{Phase IV: On-Chain Batch Admission}\label{sec:user_identity_auth}

This phase settles one admission batch on-chain using a single succinct proof and records the resulting admission anchors for later account provisioning. It proceeds through batch submission, verifier-side checks, and contract execution as described below.

\subsubsection{Batch Submission}
The PBS constructs the compressed batch input list $X = \{ {\mathsf{hash_{PHC}}}_i, {\mathsf{pk_{seed}}}_i \}_{i=1}^N$ from the admitted users' registration credentials. \revblue{The same ordered batch metadata is already bound into the Fiat--Shamir context during Phase~II, so the submitted $X$ must match the transcript-consistent batch certified by $\pi_{zk}$.} The PBS then submits $(\pi_{zk}, X, \mathcal{I}_{acc,N}, \overline{T}_{N+1}, \mathcal{I}_{acc,N+1})$ to the Verifier Contract.

\subsubsection{Verifier-Side Checks}
The contract first verifies the succinct proof and then checks that the PHC hashes in $X$ have not been recorded previously. If both checks succeed, it stores the PHC hashes and seed public keys as admission anchors, thereby creating a publicly verifiable record of admitted seed identities while preventing duplicate registration without revealing the underlying PHCs.

Algorithm~\ref{alg:uiasc} summarizes the on-chain admission logic.

\begin{algorithm}[!t]
\caption{Verifier Contract}
\label{alg:uiasc}
\begin{algorithmic}[1]
\STATE \textbf{Input:} \ensuremath{\mathcal{I}_{acc,N}, \overline{T}_{N+1},\pi_{zk}, X, \mathcal{I}_{acc,N+1}}
\STATE \ensuremath{\mathsf{stmt} \gets (\mathcal{I}_{acc,N}, \overline{T}_{N+1},\mathcal{I}_{acc,N+1})}
\STATE \ensuremath{validation \gets \mathsf{zkSNARK.VERIFY}(\mathsf{vk}, \mathsf{stmt}, \pi_{zk})}
\IF{\ensuremath{validation \neq 1}}
	    \STATE \textbf{return} ``Verification failed''
\ENDIF
\STATE \ensuremath{(\mathsf{hashes_{PHC}},\mathsf{pks_{seed}}) \gets X}
\STATE \ensuremath{stored \gets \mathsf{Traverse}(\mathsf{hashes_{PHC}})}
\IF{\ensuremath{stored = 1}}
    \STATE \textbf{return} ``Already registered''
\ENDIF
\STATE \text{Store}(\ensuremath{\mathsf{hashes_{PHC}}}, \ensuremath{\mathsf{pks_{seed}}})
\STATE \textbf{return} ``Verification succeeded and admission allowed''
\end{algorithmic}
\end{algorithm}

% ---------------------------------------------------------------------------------
% Phase V
% ---------------------------------------------------------------------------------
\subsection{Phase V: Anonymous Account Provisioning}\label{sec:soul_account_creation}

This final phase activates an admitted user's operational on-chain identity, termed a \emph{Soul Account} in ZK-AMS, after admission has been settled. It proceeds through user-side provisioning, registry-side verification, and one-to-one account activation.

\subsubsection{Provisioning Objective}
Admission and account provisioning are intentionally separated: a user is first admitted through batch settlement and then binds a fresh Soul Account anonymously.

\subsubsection{User-Side Provisioning}
\revblue{To provision a Soul Account, the user generates a new account key pair $\left( {\mathsf{addr_{soul}}},{\mathsf{pk_{soul}}},{\mathsf{sk_{soul}}} \right)$ and selects a ring $L$ of admitted seed public keys from the on-chain registry, including its own $\mathsf{pk_{seed}}$. It then computes the MLSAGS key image $y_0 \leftarrow \mathsf{sk_{seed}} \cdot \mathsf{H_p}(\mathsf{pk_{seed}})$ and signs the message $\mathsf{addr_{soul}}$ using the admitted seed key as $\sigma \leftarrow \mathsf{LRS.Sign}(1^k, 1^n, \mathsf{addr_{soul}}, L, \mathsf{sk_{seed}})$.}
\revblue{The resulting signature can be written as $\sigma=(y_0,\ldots)$, so the key image acts as a one-time provisioning tag while the signer remains hidden inside the ring. The user then submits $(\sigma, L, \mathsf{addr_{soul}})$ to the Soul Registry Contract. Appendix~\ref{supp:lrs} summarizes the MLSAGS interface and key-image mechanism used in this phase.}

\subsubsection{Registry-Side Verification}
The contract checks that all public keys in $L$ are valid admitted seed public keys, verifies the linkable ring signature on $\mathsf{addr_{soul}}$, and rejects previously used key images. \revblue{Because MLSAGS linkability requires repeated signatures under the same admitted seed secret key to induce the same key image, this uniqueness check enforces that one admitted seed identity can activate at most one Soul Account without revealing which ring member performed the provisioning.} If all checks pass, it records $(y_0,\mathsf{addr_{soul}})$ and finalizes provisioning. Algorithm~\ref{alg:usacsc} summarizes the contract logic.

\begin{algorithm}[!t]
\caption{Soul Registry Contract}
\label{alg:usacsc}
\begin{algorithmic}[1]
\STATE \textbf{Input:} \ensuremath{\sigma, L, \mathsf{addr_{soul}}}
\STATE \ensuremath{\mathsf{stored} \gets \mathsf{Traverse}(L)}
\IF{\ensuremath{\mathsf{stored} = 0}}
\STATE \textbf{return} ``The ring of public keys is invalid''
\ENDIF
\STATE \ensuremath{\mathsf{validation} \gets \mathsf{LRS.Verify}(1^k, 1^n, \mathsf{addr_{soul}}, L, \sigma)}
\IF{\ensuremath{\mathsf{validation} \neq 1}}
\STATE \textbf{return} ``Verification failed''
\ENDIF
\STATE \ensuremath{y_0 \gets \revblue{\mathsf{KeyImage}(\sigma)}}
\STATE \ensuremath{\mathsf{stored} \gets \mathsf{Traverse}(y_0)}
\IF{\ensuremath{\mathsf{stored} \neq 0}}
\STATE \textbf{return} ``Already provisioned''
\ENDIF
\STATE \text{Store}(\ensuremath{y_0}, \ensuremath{\mathsf{addr_{soul}}})
\STATE \textbf{return} ``Provisioning succeeded''
\end{algorithmic}
\end{algorithm}

Thus, ZK-AMS supports account-level anonymity while still enforcing one admitted seed identity to at most one operational Soul Account.

%=====================================================================
\section{Security and Operational Risk Analysis}\label{securityAnalysis}
%=====================================================================

ZK-AMS combines a confidential off-chain batching plane with a public on-chain settlement plane. Because its contribution is a principled systems integration rather than a new primitive, we analyze security at the \emph{system} level: the propositions below state which guarantee is obtained from which primitive or trust assumption, while primitive-level proofs remain those of the cited constructions. Longer proof sketches and explicit assumption-to-guarantee mappings for G1--G6 are deferred to Appendix~\ref{supp:security}. Table~\ref{tab:threat_model} summarizes the main threats, mitigations, and assumptions.

% --- Table start ---
\begin{table*}[!t]
	\caption{System-level threat model of ZK-AMS}
	\label{tab:threat_model}
	\centering
	\renewcommand{\arraystretch}{1.25}
	
	\begin{tabularx}{\textwidth}{
			L{0.6}
			L{1.6}
			L{0.95}
			L{0.6}
		}
		\toprule
		\textbf{Threat} & \textbf{Mitigation (Mechanism)} & \textbf{Assumption} & \textbf{Enforced / Verified at} \\
		\midrule
		
		\multicolumn{4}{l}{\textit{\textbf{Category I: Admission integrity \& uniqueness enforcement}}} \\
		Forged admission &
		The admission relation checks PHC integrity, issuer signature validity, and holder-key ownership; batch settlement succeeds only if the succinct proof verifies. &
		Signature unforgeability; zkSNARK soundness. &
		On-chain: Verifier Contract \\
		\addlinespace[2pt]
		
		Duplicate admission (credential reuse) &
		The Verifier Contract stores and checks the uniqueness of $\mathsf{hash_{PHC}}$; the proof binds each admitted batch entry to the submitted credential hash. &
		Hash collision resistance; issuer-side uniqueness of PHCs. &
		On-chain: Verifier Contract \\
		\addlinespace[2pt]
		
		Double provisioning &
		The Soul Registry Contract records the key image $y_0$ and rejects repeated use of the same admitted seed identity for multiple accounts. &
		MLSAGS linkability (unique key image per seed key). &
		On-chain: Soul Registry Contract \\
		\addlinespace[2pt]
		
		\midrule
		\multicolumn{4}{l}{\textit{\textbf{Category II: Privacy \& confidentiality preservation}}} \\
		Deanonymization / account linkage &
		MLSAGS signer ambiguity hides which admitted seed identity provisioned the Soul Account, while the key image reveals only duplicate-use information. &
		MLSAGS anonymity and linkability. &
		Publicly verifiable cryptographic guarantee \\
		\addlinespace[2pt]
		
		Witness leakage to PBS &
		Folding is evaluated over MKHE ciphertexts; the PBS sees only ciphertexts, commitments, and public transcript state. Final decryption reconstructs only a randomized folded batch output, which hides individual witnesses unless the PBS coalition leaves only one non-colluding participant in the batch. &
		MKHE semantic security before share release; at least two batch participants remain non-colluding with the PBS in a finalized batch. &
		Off-chain confidential batching layer \\
		\addlinespace[2pt]
		
		\midrule
		\multicolumn{4}{l}{\textit{\textbf{Category III: Operational resilience \& availability}}} \\
		PBS tampering / transcript inconsistency &
		Fiat--Shamir binds each folding step to the public transcript; a transcript-inconsistent batch causes final proof verification failure. &
		Fiat--Shamir heuristic (RO model); Nova/NIFS consistency. &
		Off-chain transcript + On-chain Verifier Contract \\
		\addlinespace[2pt]
		
		Faulty or missing decryption shares &
		In the current prototype, a batch cannot be finalized without the required share set; malformed or missing shares therefore cause batch stalling or exclusion and re-queuing under an off-chain policy. Stronger handling would require a verifiable share-validation or accountable-abort layer. &
		Availability of a sufficient share set; no publicly verifiable share-correctness or identifiable-abort layer in the current prototype. &
		Off-chain coordination policy \\
		\addlinespace[2pt]
		
		PBS delay / censorship / batch stalling &
		The PBS role is permissionless; after timeout, another node may finalize from the same public transcript, and affected requests can be re-queued without invalidating prior published state. &
		At least one live PBS; eventual transcript availability. &
		Operational policy layer \\
		\addlinespace[2pt]
		
		Off-chain data unavailability &
		Intermediate artifacts are stored on IPFS/IPNS and can be replicated by multiple nodes; only compact admission anchors are required on-chain for safety. &
		Eventual availability of off-chain storage. &
		Off-chain transcript layer \\
		\bottomrule
	\end{tabularx}
	
\end{table*}
% --- Table end ---

\subsection{System Assumptions and Security Goals}

\subsubsection{Assumptions}
Our analysis relies on the following assumptions.

\begin{itemize}
	\item \textbf{A1) Credential trust root:}
	Recognized issuers generate valid PHCs, and each valid PHC corresponds to at most one real-world person.
	
	\item \textbf{A2) Primitive security:}
	The adopted primitives satisfy their standard properties: issuer/holder signature unforgeability, collision-resistant hashing, zkSNARK soundness and zero knowledge, Nova/NIFS consistency in the Fiat--Shamir (random-oracle) setting, MLSAGS correctness, unforgeability, anonymity, and linkability, and BGV-style MKHE semantic security.
	
	\item \textbf{A3) Blockchain correctness:}
	The blockchain executes smart contracts faithfully and preserves the integrity of admitted state once transactions are confirmed.
	
	\item \textbf{A4) Eventual transcript availability:}
	The off-chain transcript layer (IPFS/IPNS) eventually makes published objects retrievable, although temporary delay may occur.

	\item \textbf{A5) Confidentiality coalition bound in the current prototype:}
	The current MKHE batch-materialization step uses all-party share release over the finalized batch key set and does not implement a separate tunable $t$-out-of-$N$ threshold-decryption layer. Accordingly, individual-witness confidentiality is claimed only when at least two participants remain outside the PBS coalition. If the PBS colludes with all but one finalized participant, the coalition may isolate that participant's folded plaintext contribution from the revealed accumulator after share release; this is weaker than direct recovery of the participant's entire original witness, but it is already outside the confidentiality claim of the current prototype.
\end{itemize}

\subsubsection{Security goals}
ZK-AMS aims to provide the following system-level guarantees:

\begin{itemize}
	\item \textbf{G1) Admission integrity:} only users with valid PHCs and corresponding holder keys can be admitted.
	\item \textbf{G2) Duplicate-admission resistance:} the same PHC cannot be admitted multiple times.
	\item \textbf{G3) Credibly anonymous provisioning:} an admitted user can provision an anonymous Soul Account without revealing which admitted seed identity they own.
	\item \textbf{G4) Single-account enforcement:} one admitted seed identity can provision at most one active Soul Account.
	\item \textbf{G5) Witness confidentiality during batching:} the PBS should not learn any individual user witness from the off-chain batching workflow.
	\item \textbf{G6) Batch-consistency and settlement integrity:} a tampered or transcript-inconsistent batch should fail on-chain verification.
\end{itemize}

\subsubsection{Adversarial model}
We consider PPT adversaries including malicious users attempting forged or repeated admission, a malicious PBS attempting tampering, delay, or selective suppression, malicious participants withholding or corrupting decryption shares, and external observers trying to link admitted credentials to provisioned accounts. The PBS is not trusted. Correctness claims allow arbitrary collusion subject to A1--A4. Confidentiality adds the coalition bound in A5: if the PBS colludes with all but one participant in a finalized batch, the remaining participant's folded plaintext contribution to the revealed accumulator may be isolated after share release, so the paper does not claim confidentiality in that case.

\subsection{Admission Integrity and Provisioning Correctness}

The following propositions formalize G1--G4 as system-level consequences of A1--A3.

\begin{proposition}[Admission integrity]
Under Assumptions A1--A3, if a PPT adversary causes the Verifier Contract to accept an admission request for which it does not know a valid PHC and the corresponding holder key, then one can construct a PPT adversary that breaks either issuer/holder signature unforgeability or zkSNARK soundness with non-negligible probability.
\end{proposition}

\begin{proof}[Proof sketch]
Any admitted batch must pass the on-chain zkSNARK verification for the committed admission relation. If the accepted batch encodes a false admission statement, then the adversary breaks zkSNARK soundness. Otherwise, the accepted statement must satisfy the embedded checks for PHC validity, issuer certification, and holder-key ownership; fabricating such a statement without a valid credential and holder key yields a forgery against the corresponding signature component. Hence a successful forged admission reduces to one of these primitive failures.
\end{proof}

\begin{proposition}[Duplicate-admission resistance]
Under Assumptions A1--A3, if a PPT adversary causes the same real-world identity anchor to be admitted twice without triggering the duplicate-registration check, then one can construct a PPT adversary that breaks collision resistance of $\mathsf{hash_{PHC}}$, breaks zkSNARK soundness, or violates the issuer-side uniqueness assumption in A1.
\end{proposition}

\begin{proof}[Proof sketch]
The Verifier Contract records admitted $\mathsf{hash_{PHC}}$ values on-chain and rejects repeated anchors atomically during settlement. A successful duplicate admission therefore implies either a collision in $\mathsf{hash_{PHC}}$, a false accepting proof for an inconsistent batch, or violation of issuer-side uniqueness in A1.
\end{proof}

\begin{proposition}[Anonymous yet enforceable account provisioning]
Under Assumptions A2--A3, any Soul Account accepted by the Soul Registry is authorized by some admitted seed identity, and any PPT adversary's advantage in identifying which admitted seed identity produced the accepted provisioning signature is bounded by its advantage against MLSAGS anonymity.
\end{proposition}

\begin{proof}[Proof sketch]
The Soul Registry accepts provisioning only if the submitted MLSAGS signature verifies against a ring of admitted seed public keys. Thus, acceptance without any admitted signer contradicts MLSAGS correctness or unforgeability. Given that some admitted signer exists, any non-negligible deanonymization advantage transfers directly to an adversary against MLSAGS anonymity.
\end{proof}

\begin{proposition}[Single-account enforcement]
Under Assumption A2, if a PPT adversary causes two distinct Soul Accounts derived from the same admitted seed identity to be accepted without reusing the same recorded key image, then one can construct a PPT adversary that breaks MLSAGS linkability.
\end{proposition}

\begin{proof}[Proof sketch]
The registry records the MLSAGS key image $y_0$ after successful provisioning and rejects later attempts carrying the same image. Since MLSAGS linkability requires all signatures under a fixed seed secret key to induce the same key image, accepting two distinct provisionings from the same admitted seed identity without repeated image use would violate MLSAGS linkability.
\end{proof}

\subsection{Privacy and Confidentiality Preservation}

The next two propositions capture G5--G6. For witness confidentiality, the operative condition is A5: the current prototype does not expose a tunable threshold-decryption parameter, so the guarantee is stated only for batches that leave at least two participants outside the PBS coalition.

\begin{proposition}[Witness confidentiality against the PBS]
Fix a finalized batch of size $N$. Under Assumptions A2 and A5, any coalition consisting of the PBS and at most $N-2$ batch participants cannot recover any specific remaining participant's underlying witness from the transcript except with negligible probability beyond what is implied by the public commitments and the intentionally revealed folded batch output.
\end{proposition}

\begin{proof}[Proof sketch]
Before share release, a hybrid argument replacing honest users' encrypted witnesses with encryptions of dummy values yields a distinguisher against BGV-style MKHE semantic security if the adversary can still recover an individual witness. After share release, the coalition learns the final folded plaintext accumulator. By A5, when at least two participants remain outside the PBS coalition, subtracting the colluders' known contributions reveals only a public linear combination of at least two unknown witness-dependent terms, not an isolated term attributable to any one remaining participant. If instead only one participant remains outside the coalition, then that participant's folded plaintext contribution to the revealed accumulator may be isolated after share release. This is weaker than direct recovery of the participant's entire original witness, but it is already outside the confidentiality claim of the current prototype.
\end{proof}

\begin{proposition}[Batch-consistency and settlement integrity]
Under Assumptions A2--A3, in the Fiat--Shamir random-oracle setting, if a PPT adversary tampers with the batch transcript or submits transcript-inconsistent batch material yet still causes the Verifier Contract to accept, then one can construct a PPT adversary that breaks either Nova/NIFS consistency or zkSNARK soundness with non-negligible probability.
\end{proposition}

\begin{proof}[Proof sketch]
Each folding step derives its challenge from the public transcript through Fiat--Shamir, so any change to commitments, ordering, or batch metadata changes the derived challenge unless the adversary breaks Nova/NIFS transcript consistency in the random-oracle setting. The final proof circuit checks that the submitted accumulated instance is transcript-consistent and satisfiable. Thus an accepted tampered batch implies either a Nova/NIFS consistency break or a false accepting zkSNARK proof.
\end{proof}

\subsection{Operational Resilience and Limitations}

The propositions above cover correctness and confidentiality under A1--A5. Deployment-level risks remain.

\subsubsection{Faulty or missing decryption shares}
The current prototype assumes that a sufficient set of users eventually provides the required decryption shares. If a participant withholds or malforms its share, the batch cannot be finalized from the current transcript state; even a single malicious participant can therefore stall the batch it joins. The present recovery rule is off-chain: users with missing or invalid shares are excluded after timeout and the affected requests are re-queued. Thus the issue is a liveness bottleneck rather than a safety failure of already confirmed on-chain state, but the prototype does not yet provide cryptographic accountability for the abort.

A production-grade design would strengthen this stage with publicly verifiable share validation or accountable abort. Representative directions include PVSS~\cite{StadlerPVSS1996} and MPC with identifiable abort~\cite{IshaiOstrovskyZikas2014IdentifiableAbort}. Integrating such mechanisms into batch materialization could support stronger exclusion and recovery policies than the current re-queue rule, but this is outside the scope of the prototype.

\paragraph*{Other deployment limitations.}
Because the PBS role is permissionless, another node may take over after timeout from the same public transcript; this limits monopolization but does not eliminate temporary censorship or latency inflation. Likewise, off-chain transcript artifacts are replicated through IPFS/IPNS, so confirmed on-chain state remains trustworthy even though persistent transcript unavailability can delay completion. Batch settlement is also exposed to transaction ordering and mempool races: atomic duplicate checks prevent forged admission under successful contract verification, but conflicting anchors can still cause revert or re-queue cost. Finally, ZK-AMS does not yet integrate post-admission PHC revocation, expiry, or reissuance, nor does it provide native Soul-Account recovery after key loss. Possible extensions include epoch-based re-admission, revocation registries, social recovery, recovery-authorized reprovisioning, private relay submission, encrypted mempools, or commit-then-reveal style anchoring.

\subsubsection{Current scope of guarantees}
ZK-AMS provides admission integrity, duplicate-admission resistance, witness confidentiality during batching, and anonymous-but-enforceable provisioning under the stated assumptions. Its current limits are equally clear: uniqueness still depends on PHC issuance, confidentiality does not cover the coalition ``PBS + all but one participant'' because that coalition may isolate the last remaining participant's folded plaintext contribution after share release, liveness against malicious share withholding or strategic PBS behavior is only partially addressed through off-chain recovery, transaction ordering can still impose denial or re-queue cost, and neither PHC lifecycle management nor Soul-Account recovery is yet integrated. These limits do not invalidate the correctness of admitted batches, but they do affect privacy margin, availability, and long-term deployability.

%==================================================================
\section{Implementation and Experimental Evaluation} \label{experiment}
%==================================================================

We evaluate ZK-AMS as a deployable admission infrastructure for Web 3.0 platforms, where onboarding must remain privacy-preserving, Sybil-resistant, and operationally predictable under bursty demand. Compared with the preliminary conference version, the present implementation includes the full end-to-end workflow of confidential off-chain batching, permissionless batch finalization, on-chain batch admission, and anonymous account provisioning. Our evaluation aims to answer the following questions:

\begin{itemize}
	\item \textbf{EQ1 (Scalable settlement)}: Does recursive proof aggregation keep on-chain verification cost stable as the admission batch size $N$ grows, thereby enabling cost-predictable batch settlement?
	\item \textbf{EQ2 (Admission-settlement latency)}: What is the admission-settlement processing latency, excluding batch-formation waiting time, and how is it distributed across client-side preparation, off-chain transcript operations, PBS-side finalization, and on-chain settlement?
	\item \textbf{EQ3 (Configuration trade-offs)}: How do batch size $N$ and ring size $L$ affect throughput, latency, and per-user cost, and what parameter settings are practically attractive?
\end{itemize}

\subsection{Implementation Setup and Metrics}

\textbf{Testbed}: We deploy the Verifier Contract and Soul Registry Contract on a private six-node Ethereum network running Clique proof-of-authority (PoA) with Go-Ethereum (Geth)~\cite{go-ethereum} on Alibaba Cloud. The block period is 12\,s. We use private PoA to stabilize block production and contract execution so that workflow cost is repeatable under controlled load. For throughput translation, we set the block gas limit to 60M and use the EIP-1559 target gas (30M) as the nominal per-block gas budget. Thus, the reported gas costs are more portable across EVM-compatible deployments than the wall-clock confirmation delay, which would differ on Ethereum mainnet or L2 networks.

\textbf{Roles and hardware}: We evaluate two representative roles on the same hardware platform, a Mac mini with Apple M4 and 16\,GB RAM. One device serves as the \emph{Permissionless Batch Submitter (PBS)} for batch fusion and proof finalization, and the other as the \emph{user client} for credential processing, confidential-folding participation, and MLSAGS signing.

\textbf{Software stack}: We implement the contracts in Solidity on the private Ethereum testbed; the proof layer uses Groth16~\cite{groth2016size} over BN254, recursive aggregation uses Arkworks~\cite{arkworks} and Sonobe~\cite{sonobe}, and the confidential batching layer uses a BGV-style MKHE implementation with IPFS/IPNS as the transcript substrate.

\textbf{Metrics}: We report PBS-side proving cost and throughput, on-chain settlement gas, admission-settlement processing latency to first confirmation (excluding batch-formation waiting time) and its breakdown, and correctness/operational metrics including decryption success and confidential-folding overhead. Unless otherwise stated, each measurement is repeated $n_{\mathrm{rep}}=10$ times and we report the mean.

\textbf{Evaluation scope}: We separate three result regimes. First, full admission-settlement path measurements (Phases~I--IV in Section~\ref{detaildesign}), including confidential folding, distributed decryption, batch finalization, and on-chain settlement, are reported up to $N=128$, the largest batch size repeatedly executed on the current hardware with the chosen MKHE parameterization. Second, larger-$N$ results up to $N=1100$ isolate the recursive proving and settlement path after batch compression and should therefore be read as proof-compression scaling evidence rather than full confidential-pipeline measurements. Third, Phase~V account-provisioning overhead is reported separately through MLSAGS signing and verification measurements rather than folded into the admission-settlement latency tables.

\revred{\textbf{Baseline and comparison scope}: We compare ZK-AMS experimentally against the original non-recursive zkBID design~\cite{zkbid}, which removes recursive batch compression while preserving the same admission semantics. This implemented baseline isolates the impact of recursive aggregation on proving time, on-chain verification cost, and amortized admission efficiency. We also report account-provisioning overhead through MLSAGS verification. For clarity, zkBID is the only implemented system-level baseline in this section.}

\subsection{MKHE Correctness and Noise Robustness}

We first test whether the confidential batching layer remains correct at representative large-batch checkpoints. Table~\ref{tab:mkhe_correctness} reports batch size, effective multiplication depth, repeated trials, decryption failures, and remaining noise margin, where margin (bits) denotes the modulus-chain slack to the first decryption failure.

\begin{table}[!t]
	\centering
	\caption{MKHE correctness in repeated large-batch checkpoints.}
	\label{tab:mkhe_correctness}
	\scriptsize
	\setlength{\tabcolsep}{3pt}
	\renewcommand{\arraystretch}{1.03}
	\begin{tabular}{@{}lcccc@{}}
		\hline
		\textbf{$N$} & \textbf{Depth} & \textbf{Trials} & \textbf{Fails} & \textbf{Margin (bits)} \\ \hline
		16 & 8  & 10 & 0 & 67 \\
		32 & 10 & 10 & 0 & 67 \\
		64 & 12 & 10 & 0 & 67 \\ \hline
	\end{tabular}
\end{table}

Table~\ref{tab:mkhe_correctness} shows zero decryption failures in these repeated checkpoints up to $N=64$ and a 67-bit modulus-chain slack throughout them. Table~\ref{tab:latency_breakdown} further includes successful full admission-settlement measurements at $N=128$, where confidential folding reaches 1615.85\,s. We therefore report full admission-settlement measurements up to $N=128$ and treat larger-$N$ plots in EQ1 as recursion-scaling evidence rather than direct measurements of the full confidential workflow. \revred{This overhead should be read as the cost of preserving witness confidentiality during batching, not as the source of the verifier-gas benefit, which comes from the recursive settlement path itself.}

\subsection{EQ1: Scalable Settlement and Predictable On-Chain Cost}

We next examine how the admission-settlement layer scales with batch size $N$. Unless stated otherwise, this subsection focuses on the recursive proving and settlement path after batch compression rather than the full confidential admission pipeline.

\subsubsection{PBS-side proving scalability}

The PBS pipeline is the main off-chain throughput bottleneck. In the non-recursive baseline, proving cost grows more rapidly with batch size because more admission workload is pushed directly into proof generation, whereas ZK-AMS operates on a batch-compressed statement obtained through recursive folding.

Fig.~\ref{ComparisonOfArithmeticCircuit} shows that the baseline circuit size grows almost linearly, from 5.16M constraints at $N=400$ to 14.19M at $N=1100$, whereas the recursive design stays fixed at 9.894M constraints. The non-recursive baseline is still smaller for $N \le 700$, but the crossover occurs between $N=700$ and $N=800$; from $N=800$ onward, the recursive design yields the smaller final proof circuit. In other words, recursive aggregation pays a fixed circuit footprint for better large-batch scaling.

Fig.~\ref{ComparisonOfZKPGeneration} shows that the baseline proving time grows from 69.27\,s at $N=400$ to 320.57\,s at $N=1100$, while ZK-AMS remains in a narrow band of about 140--150\,s. The baseline is faster for $N=400$--$700$, but ZK-AMS becomes faster from $N=800$ onward and widens the gap at larger batch sizes. In throughput terms, the baseline drops from 5.77 to 3.43 users/s, whereas ZK-AMS rises from 2.66 users/s to a peak of 7.76 users/s at $N=1000$ and remains above 7.7 users/s at $N=1100$. This shows that the recursive pipeline pays a fixed aggregation overhead at moderate batch sizes but scales better once the batch is large enough.

\subsubsection{On-chain cost predictability}

Fig.~\ref{ComparisonOfGasCost} reports \emph{batch-proof verification gas} only. The recursive ZK-AMS curve is shown for $2 \le N \le 256$ because recursive batch verification becomes meaningful only once at least one folding step exists; thus $N=1$ is omitted from the recursive curve and kept only as the single-user reference for the non-batched baseline. The baseline corresponds to verifying $N$ independent Groth16 admission proofs on-chain, so its verifier gas is computed as $C_{\mathrm{ind}}(N)\approx N\times 250.44$K from our measured single-proof verifier cost at $N=1$. Under this construction, the baseline rises from 250.44K at $N=1$ to 11.20M at $N=256$, whereas the recursive verifier gas stays near 791K across the recursive range. The crossover occurs between $N=8$ and $N=16$: the baseline is cheaper up to $N=8$, but from $N=16$ onward the recursive design incurs lower verifier-side gas. In short, recursive batch-proof verification keeps verifier gas approximately flat as batch size increases.

The full \emph{batch-admission settlement gas} further includes the contract-side admission overhead for storing and checking the submitted batch metadata. Combining the recursive batch-proof verification cost from Fig.~\ref{ComparisonOfGasCost} with the additional contract-side overhead of about 143K gas observed in Fig.~\ref{fig:performance_graph}(d), we estimate the full batch-admission settlement cost as about $791$K $+143$K $\approx 934$K gas.

At scale, the dominant chain-side cost shifts from batch admission settlement to per-user account provisioning. Using $C_{\mathrm{user}}(N)\approx C_{\mathrm{MLSAGS}} + \frac{C_{\mathrm{settle}}}{N}$ with $C_{\mathrm{settle}}\approx 934$K gas, the amortized batch-admission settlement contribution falls from about 467K gas at $N=2$ to about 3.6K gas at $N=256$; for sufficiently large batches, the per-user chain-side cost is therefore dominated by MLSAGS account provisioning rather than recursive batch settlement.

\subsection{EQ2: Admission-Settlement Latency Breakdown}

We next decompose the admission-settlement processing latency to first confirmation, excluding the batch-formation waiting time $T_{wait}(N)$ from Eq.~\eqref{equ:batch_latency}, into client-side preparation, IPFS/IPNS access, confidential folding and share generation, PBS-side fusion and proof finalization, and on-chain inclusion/execution delay.

\begin{table}[!t]
	\captionsetup{justification=centering}
	\centering
	\caption{Admission-settlement processing latency breakdown to first confirmation (excluding batch-formation waiting time).}
	\label{tab:latency_breakdown}
	\scriptsize
	\setlength{\tabcolsep}{3pt}
	\renewcommand{\arraystretch}{1.03}
	\makebox[\columnwidth][c]{%
	\begin{tabular}{ccccccc}
		\hline
		\textbf{$N$} & \textbf{Prep} & \textbf{IPFS} & \textbf{Fold} & \textbf{PBS} & \textbf{Chain} & \textbf{Total} \\
		& \textbf{(s)} & \textbf{(s)} & \textbf{(s)} & \textbf{(s)} & \textbf{(s)} & \textbf{(s)} \\ \hline
			1  & 0.51 & 0.04 & 0.13   & 125.85 & 60.48 & 187.00 \\
			16 & 1.56 & 0.06 & 86.41  & 124.24 & 58.88 & 271.14 \\
			32 & 1.71 & 0.06 & 218.78 & 118.67 & 60.56 & 399.78 \\
			64 & 1.76 & 0.07 & 610.91 & 122.89 & 58.50 & 794.14 \\
			128 & 1.78 & 0.05 & 1615.85 & 125.43 & 52.78 & 1795.89 \\ \hline
		\end{tabular}
	}
\end{table}

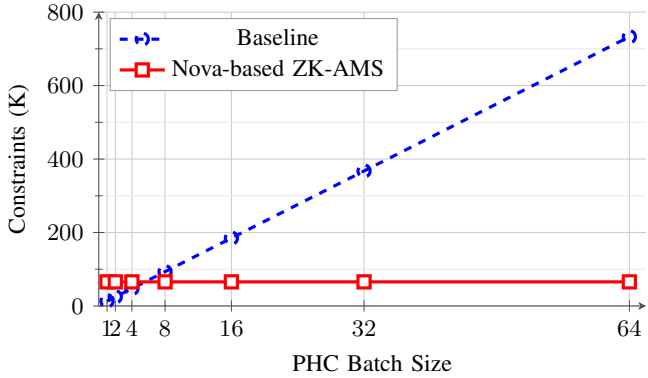
\begin{figure}[!t]
	\centering
	\resizebox{0.9\linewidth}{!}{\input{fig7_constraints.tikz}}
	\caption{Circuit constraint counts versus admission batch size $N$.}
	\label{ComparisonOfArithmeticCircuit}
	\vspace{-0.2cm}
\end{figure}

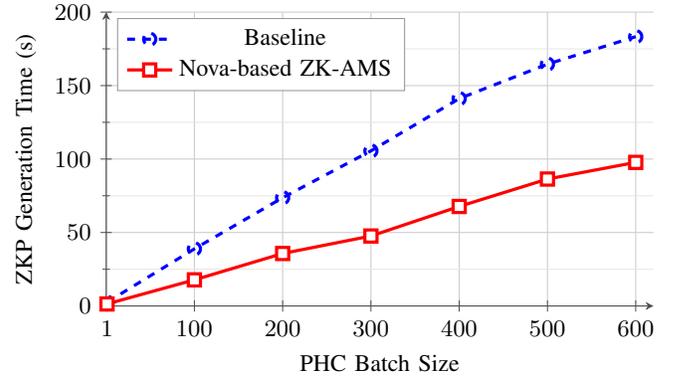
\begin{figure}[!t]
	\centering
	\resizebox{0.9\linewidth}{!}{\input{fig8_proving_time.tikz}}
	\caption{PBS proof-generation time versus admission batch size $N$.}
	\label{ComparisonOfZKPGeneration}
	\vspace{-0.4cm}
\end{figure}

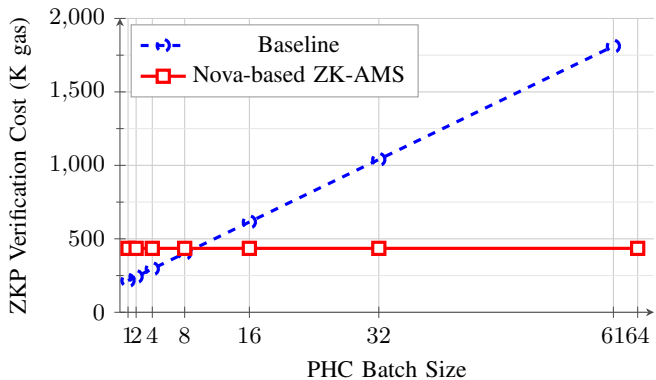
\begin{figure}[!t]
	\centering
	\resizebox{0.9\linewidth}{!}{\input{fig9_zkp_verify_gas.tikz}}
	\caption{On-chain gas cost for batch-proof verification versus batch size $N$.}
	\label{ComparisonOfGasCost}
	\vspace{-0.2cm}
\end{figure}

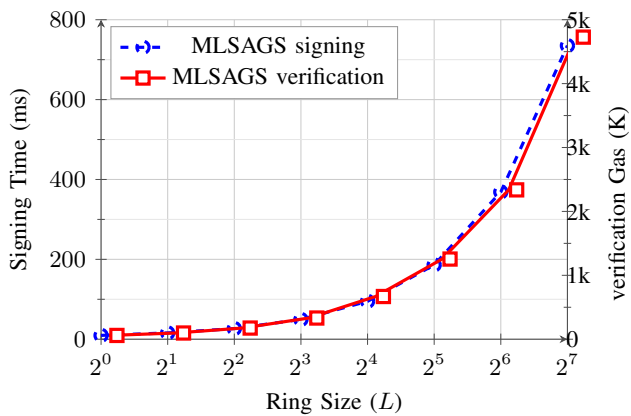
\begin{figure}[!t]
	\centering
	\resizebox{0.9\linewidth}{!}{\input{fig10_ringsig.tikz}}
	\caption{MLSAGS signing time (client) and on-chain verification gas versus ring size $L$.}
	\label{RingSigData}
	\vspace{-0.2cm}
\end{figure}

Table~\ref{tab:latency_breakdown} shows a clear separation between nearly batch-invariant and batch-sensitive components. The $N=1$ row serves only as a single-request reference point rather than a recursive-batching measurement. Across the measured batched runs, client-side preparation stays within 1.56--1.78\,s, IPFS/IPNS access within 0.05--0.07\,s, the observed chain-side 1-confirm delay within 52.78--60.56\,s, and PBS-side finalization roughly flat at 118.67--125.43\,s. Notably, the observed chain delay is substantially larger than the nominal 12\,s block period of our private PoA testbed, so wall-clock confirmation should be interpreted separately from gas portability. By contrast, confidential folding grows from 0.13\,s at $N=1$ to 1615.85\,s at $N=128$ and dominates the large-batch latency. As a result, the measured admission-settlement processing latency rises from 187.00\,s at $N=1$ to 1795.89\,s at $N=128$. This confirms that the full confidential admission-settlement path remains executable at $N=128$ on the current hardware, while also making clear that confidential folding is the limiting stage. Within the scope of these processing-only measurements, the critical path is therefore well approximated by $T_{\mathrm{adm}}^{(1)} \approx \max_i T_{\mathrm{client},i} + T_{\mathrm{transcript}} + T_{\mathrm{pbs}} + T_{\mathrm{chain}}$, where $T_{\mathrm{transcript}}$ includes IPFS/IPNS access and share availability; the full batch latency in Eq.~\eqref{equ:batch_latency} would additionally include $T_{wait}(N)$.

This latency profile indicates that the current prototype is better suited to high-value or bursty admission windows than to interactive consumer login. Still, several stages are naturally parallel across users, including credential preparation, ciphertext generation, and transcript publication, and these steps for batch $k+1$ can overlap with PBS finalization and on-chain settlement of batch $k$. By contrast, the encrypted accumulator update is only partially parallelizable because each fold is bound to the latest public transcript through a Fiat--Shamir challenge. The most realistic near-term optimization path is therefore pipelined batch overlap plus parallel preprocessing, rather than full elimination of the sequential folding bottleneck.

\begin{figure*}[htbp]
	\centering
	\includegraphics[width=0.85\textwidth]{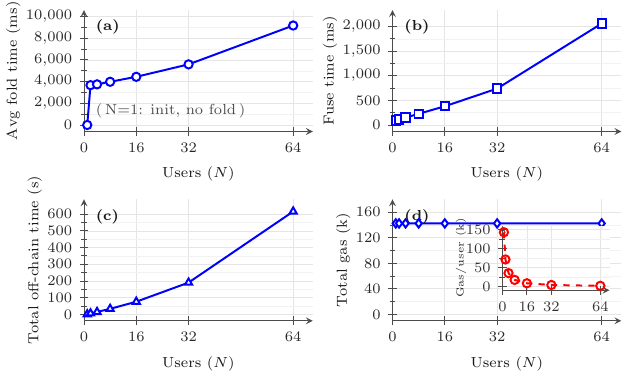}
	\caption{Scalability of confidential off-chain folding as the participant count $N$ increases:
		(a) average fold time;
		(b) fuse time;
		(c) aggregate off-chain time; and
		(d) on-chain gas overhead, with an inset showing amortized gas per admitted user.}
	\label{fig:performance_graph}
	\vspace{-0.2cm}
\end{figure*}

\subsection{EQ3: Configuration Trade-offs and System Bottlenecks}

\subsubsection{\texorpdfstring{Ring size $L$ and account-provisioning overhead}{Ring size L and account-provisioning overhead}}
The anonymity of post-admission account provisioning is controlled by the MLSAGS ring size $L$. Fig.~\ref{RingSigData} shows that both signing time and verification gas increase approximately linearly with $L$ over the range $L\in[1,256]$. Specifically, MLSAGS signing grows from 19.19\,ms at $L=1$ to 3425.11\,ms at $L=256$, while on-chain verification gas rises from 98.21K gas to 16.67M gas. This makes $L$ an explicit anonymity--cost knob. In practice, a ring size $L\in[8,16]$ remains the most practical operating region in our prototype, because it keeps client-side signing below about 221\,ms and verifier-side gas below about 1.07M while still providing a non-trivial anonymity set. This is also consistent with practical anonymity settings used in privacy-oriented systems such as Monero~\cite{monero_0_13_0_released_2018}.

\subsubsection{Decentralized confidential folding overhead}

Fig.~\ref{fig:performance_graph} reports the overhead of the confidential off-chain batching layer. As expected, average fold time, fuse time, and aggregate off-chain work all increase with $N$, while the amortized on-chain gas per admitted user decreases because the fixed batch-settlement cost is shared across more users. This trade-off is orthogonal to the verifier-gas benefit in EQ1: the latter comes from the recursive settlement path itself, whereas the confidential folding layer is the price paid for preserving witness confidentiality during batching. On the current hardware and parameter set, $N=16$--$32$ offers the clearest practical compromise between amortization and processing delay, while $N=64$--$128$ is better suited to scheduled onboarding windows in which a roughly 13--30 minute admission-settlement delay remains acceptable.

\subsubsection{PBS bottleneck profiling}

Profiling further shows that, within the PBS stage, proof finalization remains the dominant subtask for large batches. In our implementation, final proof generation accounts for 84.3\% of PBS-side latency at $N=64$, while transcript fusion accounts for 10.8\% and serialization plus transaction preparation accounts for the remaining 4.9\%. Peak memory consumption remains within the 16\,GB memory budget of the PBS device, indicating that the current PBS bottleneck is primarily compute-bound.

% \begin{figure*}[htbp]
%     \centering
% 	\input{figures/fig11_folding_performance.tikz}
%     \caption{Scalability of the decentralized privacy-preserving folding as the participant count ($N$) increases: (a) Avg fold time; (b) Fuse time; (c) Total off-chain time (aggregate load); and (d) total gas, with an inset showing amortized gas per admitted user.}
%     \label{fig:performance_graph}
% \end{figure*}

%================================================================================
\section{Related Work}\label{rw}
%================================================================================

\begin{table*}[!t]
	\centering
	\caption{Comparison of representative systems by system-level admission capabilities
		(\full~Fully supported~~\partSup~Partially supported~~\none~Not supported)}
	\label{tab:rw_capabilities}
	\resizebox{0.98\textwidth}{!}{
		\begin{tabular}{l|cccccccc|c}
			\toprule
			& \multicolumn{8}{c}{\textbf{Representative Systems}} & \textbf{Ours} \\
			\cmidrule(lr){2-9} \cmidrule(lr){10-10}
			\textbf{System-Level Capability}
			& \cite{hyperledger_indy_docs}
			& \cite{sovrin}
			& \cite{uport_archive}
			& \cite{niya2019trademap}
			& \cite{rathee2022zebra}
			& \cite{pauwels2022zkkyc}
			& \cite{weidentity}
			& \cite{PopaTSC2023Chaindiscipline}
			& \textbf{ZK-AMS}
			\\
			\midrule
			
			Credential-based privacy-preserving admission
			& \partSup & \full & \partSup & \full & \full & \full & \partSup & \partSup & \full \\
			
			Ex-ante uniqueness / Sybil-resistant admission
			& \none & \none & \none & \partSup & \partSup & \partSup & \partSup & \none & \full \\
			
			Permissionless batch admission workflow
			& \none & \none & \none & \none & \none & \none & \none & \none & \full \\
			
			Witness-confidential recursive batching
			& \none & \none & \none & \none & \none & \none & \none & \none & \partSup \\
			
			Constant-cost on-chain settlement per batch
			& \none & \none & \none & \none & \none & \none & \none & \none & \full \\
			
			Anonymous post-admission account provisioning
			& \none & \none & \none & \none & \none & \none & \none & \none & \full \\
			
			\bottomrule
		\end{tabular}
	}
	\vspace{-0.4cm}
\end{table*}

ZK-AMS lies at the intersection of three research lines: identity admission and provisioning for decentralized platforms, privacy-preserving credential verification and on-chain settlement, and scalable proof aggregation for high-concurrency verification. Existing works cover important parts of this design space, but they typically do not provide credibly anonymous admission, recursively aggregated batching with explicitly scoped witness-confidentiality guarantees, permissionless batch finalization, and anonymous account provisioning in one system.

\subsection{Web 3.0 Admission and Identity Provisioning}
Identity is increasingly treated as a platform primitive in decentralized systems, especially in DAO governance, decentralized social systems, and incentive-driven Web 3.0 applications. Existing approaches include SSI/DID/VC ecosystems such as Hyperledger Indy, Sovrin, and uPort~\cite{hyperledger_indy_docs,sovrin,uport_archive}, as well as blockchain-based identity-management and decentralized authorization frameworks~\cite{PopaTSC2023Chaindiscipline,weidentity}. These systems provide useful support for identity issuance, selective disclosure, and decentralized authorization, but they generally stop at identity management, credential presentation, or policy enforcement rather than the full onboarding path from real-world uniqueness credential to anonymous-but-enforceable platform account provisioning.

\subsection{Privacy-Preserving Credential Verification and On-Chain Settlement}
A second line of work studies privacy-preserving verification of credentials and identity attributes in blockchain environments, including KYC-oriented zk systems, anonymous credential schemes, and threshold credential mechanisms~\cite{niya2019trademap,rathee2022zebra,pauwels2022zkkyc,ShiTSC2023ThresholdABC}. These works improve privacy compared with direct disclosure, but many still rely on per-user verification or trusted coordination, so settlement cost grows with the number of joiners. ZK-AMS instead targets batch-settled platform admission with predictable settlement cost under high-concurrency onboarding.

\subsection{Recursive Aggregation and Confidential Batching}
A third line of work uses proof aggregation and recursive composition to improve blockchain scalability~\cite{ShashidharaSecurityPrivacy2024ZKP,LiTITS2023AggregatedZKPPlatooning}. These techniques reduce on-chain verification overhead by compressing multiple statements or execution steps into a single proof, but generic aggregation mechanisms do not by themselves solve the admission problem considered here. In particular, they usually do not address permissionless batching with explicitly stated witness-confidentiality scope or anonymous account provisioning after admission settlement. ZK-AMS therefore treats recursive aggregation as part of a complete Web 3.0 admission workflow rather than as a standalone proof-compression primitive.

Table~\ref{tab:rw_capabilities} summarizes representative systems from the perspective of system-level admission capabilities. Compared with prior works, ZK-AMS combines credential-based anonymous admission, ex-ante uniqueness enforcement, recursively aggregated batching with explicitly scoped witness-confidentiality guarantees, constant-cost batch settlement, and anonymous post-admission provisioning. For the batching-confidentiality row, ZK-AMS is marked as partial support because the current prototype offers this property only under the explicit assumptions in Section~\ref{securityAnalysis}. This combination is the main distinction between ZK-AMS and existing identity, credential, or proof-aggregation systems.

%=====================================================================
\section{Conclusion}\label{concl}
%=====================================================================

This paper presented ZK-AMS, a credibly anonymous admission infrastructure for Web 3.0 platforms that combines zero-knowledge credential validation, confidential off-chain batch aggregation under explicit prototype assumptions, recursive proof aggregation, and anonymous Soul-Account provisioning. Our results show that recursive aggregation shifts settlement from per-user verification to approximately constant batch-level verification, yielding predictable on-chain cost while exposing the current bottlenecks in confidential-folding latency, the scoped confidentiality margin of the present share-release design, per-user ring-signature provisioning cost, and incentive-compatible handling of PBS misbehavior or share withholding.

Future work will focus on reducing latency through pipelined batching and parallel preprocessing, strengthening permissionless finalization with verifiable share handling and better incentives, supporting credential revocation or periodic re-admission, and exploring account recovery for lost Soul-Account keys.

\appendices
\input{appendix_arxiv}

\ifCLASSOPTIONcaptionsoff
  \newpage
\fi

\bibliographystyle{IEEEtran}
\bibliography{IEEEabrv, my_ref_database.bib}

\end{document}

%% file: fig7_constraints.tikz
\begin{tikzpicture}
\begin{axis}[
  zkamsAxis,
  width=\columnwidth,
  height=0.62\columnwidth,
  xlabel={PHC Batch Size},
  ylabel={Constraints (K)},
  xmin=400, xmax=1100,
  ymin=0, ymax=16000,
  xtick={400,500,600,700,800,900,1000,1100},
  legend style={at={(0.02,0.98)}, anchor=north west},
]
  \addplot[
    zkBlueLine,
    dashed,
    mark=o,
    mark options={fill=white, draw=zkBlue, solid},
  ] coordinates {
    (400,5160.400)
    (500,6450.500)
    (600,7740.600)
    (700,9030.700)
    (800,10320.800)
    (900,11610.900)
    (1000,12901.000)
    (1100,14191.100)
  };
  \addlegendentry{Baseline}

  \addplot[
    zkRedLine,
    solid,
    mark=square*,
    mark options={fill=white, draw=zkRed},
  ] coordinates {
    (400,9894.404)
    (500,9894.404)
    (600,9894.404)
    (700,9894.404)
    (800,9894.404)
    (900,9894.404)
    (1000,9894.404)
    (1100,9894.404)
  };
  \addlegendentry{ZK-AMS}
\end{axis}
\end{tikzpicture}

%% file: fig8_proving_time.tikz
\begin{tikzpicture}
\begin{axis}[
  zkamsAxis,
  width=\columnwidth,
  height=0.62\columnwidth,
  xlabel={PHC Batch Size},
  ylabel={ZKP Generation Time (s)},
  xmin=400, xmax=1100,
  ymin=0, ymax=450,
  xtick={400,500,600,700,800,900,1000,1100},
  legend style={at={(0.02,0.98)}, anchor=north west},
]
  \addplot[
    zkBlueLine,
    dashed,
    mark=o,
    mark options={fill=white, draw=zkBlue, solid},
  ] coordinates {
    (400,69.267395278)
    (500,88.235643014)
    (600,104.084401347)
    (700,141.251110014)
    (800,168.689267858)
    (900,211.093685742)
    (1000,259.940030025)
    (1100,320.571285583)
  };
  \addlegendentry{Baseline}

  \addplot[
    zkRedLine,
    solid,
    mark=square*,
    mark options={fill=white, draw=zkRed},
  ] coordinates {
    (400,142.497709046)
    (500,148.959201018)
    (600,140.689481364)
    (700,150.686470009)
    (800,147.549456837)
    (900,138.532148072)
    (1000,142.808929368)
    (1100,149.706087919)
  };
  \addlegendentry{ZK-AMS}
\end{axis}
\end{tikzpicture}

%% file: fig9_zkp_verify_gas.tikz
\begin{tikzpicture}
\begin{axis}[
  zkamsAxis,
  width=0.98\columnwidth,
  height=0.62\columnwidth,
  xlabel={PHC Batch Size},
  ylabel={ZKP Verification Cost (K gas)},
  xmin=0.8, xmax=6.2,
  ymin=0, ymax=12000,
  ytick={0,2000,4000,6000,8000,10000,12000},
  xtick={1,2,3,4,5,6},
  xticklabels={1,4,16,64,128,256},
  xticklabel style={rotate=30, anchor=east},
  legend style={at={(0.02,0.98)}, anchor=north west},
]
  \addplot[
    zkBlueLine,
    dashed,
    mark=o,
    mark options={fill=white, draw=zkBlue, solid},
  ] coordinates {
    (1,250.438)
    (1.5,293.356)
    (2,379.288)
    (2.5,550.996)
    (3,894.544)
    (3.5,1581.568)
    (4,2955.712)
    (5,5704.084)
    (6,11200.852)
  };
  \addlegendentry{Baseline}

  \addplot[
    zkRedLine,
    solid,
    mark=square*,
    mark options={fill=white, draw=zkRed,solid},
  ] coordinates {
    (1.5,791.042)
    (2,791.846)
    (2.5,791.846)
    (3,791.834)
    (3.5,791.822)
    (4,791.810)
    (5,791.822)
    (6,791.822)
  };
  \addlegendentry{ZK-AMS}
\end{axis}
\end{tikzpicture}

%% file: fig10_ringsig.tikz
\begin{tikzpicture}
  % Left y-axis: signing time (ms).
  \begin{axis}[
    zkamsAxis,
    name=zkamsRingAxis,
    scale only axis,
    % Keep the full tikzpicture within a single IEEE column (avoid overfull hbox).
    width=0.70\columnwidth,
    height=0.48\columnwidth,
    xlabel={Ring Size ($L$)},
    ylabel={Signing Time (ms)},
    xmode=log,
    log basis x=2,
    xmin=1, xmax=256,
    ymin=0, ymax=4000,
    ytick={0,1000,2000,3000,4000},
    xtick={1,2,4,8,16,32,64,128,256},
    legend style={at={(0.02,0.98)}, anchor=north west},
  ]
    \addplot[
      zkBlueLine,
      dashed,
      mark=o,
      mark options={fill=white, draw=zkBlue, solid},
    ] coordinates {
      (1,19.190317)
      (2,32.361154)
      (4,60.670213)
      (8,112.705363)
      (16,220.338425)
      (32,429.678908)
      (64,854.006967)
      (128,1701.921100)
      (256,3425.108692)
    };
    \addlegendentry{MLSAGS signing time}

    % Legend entry for the right-axis series (drawn in the overlay axis below).
    \addlegendimage{
      zkRedLine,
      solid,
      mark=square*,
      mark options={fill=white, draw=zkRed},
    }
    \addlegendentry{MLSAGS verification gas}
  \end{axis}

  % Right y-axis: verification gas (K gas).
  \begin{axis}[
    zkamsAxis,
    scale only axis,
    at={(zkamsRingAxis.south west)},
    anchor=south west,
    width=0.70\columnwidth,
    height=0.48\columnwidth,
    axis x line=none,
    axis y line*=right,
    axis line style={draw=black!70, -{Stealth[length=3pt,width=3pt]}},
    axis background/.style={fill=none},
    grid=none,
    xtick=\empty,
    xmode=log,
    log basis x=2,
    xmin=1, xmax=256,
    ymin=0, ymax=24000,
    ytick={0,6000,12000,18000,24000},
    yticklabels={0,6k,12k,18k,24k},
    yticklabel style={xshift=-0.6em},
    ylabel={verification Gas (K)},
    ylabel style={xshift=-0.9em},
  ]
    \addplot[
      zkRedLine,
      solid,
      mark=square*,
      mark options={fill=white, draw=zkRed},
    ] coordinates {
      (1,98.210600)
      (2,162.489600)
      (4,291.519800)
      (8,550.788400)
      (16,1067.471600)
      (32,2099.674400)
      (64,4171.840200)
      (128,8348.845600)
      (256,16666.245400)
    };
  \end{axis}
\end{tikzpicture}

%% file: appendix_arxiv.tex
\renewcommand{\thesubsection}{\thesection.\arabic{subsection}}
\renewcommand{\theequation}{\thesection.\arabic{equation}}
\makeatletter
\@addtoreset{equation}{section}
\makeatother

\section{Zero-Knowledge Proofs and Nova Recursive Arguments}\label{supp:nova}

% ====== (BEGIN) Copied from the main paper Appendix app:nova ======
We utilize zkSNARK \cite{zksnark} to enable efficient verification of computational integrity. A zkSNARK scheme for an arithmetic circuit relation $\mathcal{R}_C = \{(x, W) : C(x, W) = 0\}$ consists of three algorithms:
\begin{itemize}
    \item $(\mathsf{pk, vk}) \leftarrow \mathsf{zkSNARK.KEYGEN}(1^\lambda, C)$: Generates a proving key $\mathsf{pk}$ and verification key $\mathsf{vk}$ given security parameter $\lambda$ and circuit $C$.
    \item $\pi \leftarrow \mathsf{zkSNARK.PROVE}(\mathsf{pk}, x, W)$: Outputs a succinct proof $\pi$ attesting that the prover knows a witness $W$ satisfying $C(x, W)=0$.
    \item $b \leftarrow \mathsf{zkSNARK.VERIFY}(\mathsf{vk}, x, \pi)$: Outputs $b=1$ (accept) if $\pi$ is valid for instance $x$, otherwise $b=0$.
\end{itemize}
We specifically adopt Groth16 \cite{groth2016size} for its constant-size proofs and efficient on-chain verification process. However, relying on individual proofs implies that the number of proofs and the associated on-chain verification costs grow linearly with the frequency of user interactions. To alleviate this burden on the blockchain, we employ Nova \cite{nova}, a recursive argument system based on folding schemes, to aggregate multiple execution steps into a single verifiable instance, thereby enabling the verification of a batch of operations simultaneously.

Nova transforms the standard Rank-1 Constraint System (R1CS) into a \emph{relaxed} R1CS and further into a \emph{committed relaxed} R1CS to support efficient folding with zero-knowledge.
Let $A,B,C \in \mathbb{F}^{m \times \ell}$ be fixed constraint matrices.
A relaxed instance is parameterized by a slack scalar $u \in \mathbb{F}$ and an error vector $E \in \mathbb{F}^m$, together with a witness-dependent vector $Z = (W, x, u),$ where $W$ denotes the private witness and $x$ denotes the public input. The relaxed R1CS constraint is:
\begin{equation}\label{rR1CS}
(A \cdot Z)\circ (B \cdot Z) = u \cdot (C \cdot Z) + E
\end{equation}
where $\circ$ denotes the Hadamard (entry-wise) product.

Through a Non-Interactive Folding Scheme (NIFS), a Nova prover folds two relaxed instances (and their witnesses) into a single relaxed instance using a Fiat--Shamir challenge $r \in \mathbb{F}$.
Let $Z_1=(W_1,x_1,u_1)$ and $Z_2=(W_2,x_2,u_2)$ be the corresponding vectors, and define the cross-term (as in Nova's folding construction)
$$
T = (A Z_1)\circ (B Z_2) + (A Z_2)\circ (B Z_1) - u_1 \cdot (C Z_2) - u_2 \cdot (C Z_1).
$$
Then the folded slack and error are updated as
\begin{equation}\label{nova-fold-update}
\begin{aligned}
u &= u_1 + r \cdot u_2, \\
E &= E_1 + r \cdot T + r^2 \cdot E_2.
\end{aligned}
\end{equation}
By iteratively applying this folding mechanism---merging the accumulated instance representing all prior steps with the new instance of the current step---Nova achieves Incrementally Verifiable Computation (IVC), effectively compressing a potentially unbounded sequence of steps into a single verifiable instance.

To ensure zero-knowledge, Nova commits to witness-dependent components using an additively-homomorphic commitment scheme (Pedersen is a standard instantiation in Nova). In our system, to enable homomorphic evaluation under BGV MKHE, we instantiate the commitment scheme with an additively-homomorphic \emph{linear commitment} over the BGV plaintext ring, which is introduced in Section~2 of the main paper. In the committed relaxed R1CS formulation, the \emph{public instance} contains the commitments and public inputs, while the \emph{witness} contains the openings. Concretely, we denote a (public) committed relaxed instance by $\mathcal{I} \triangleq (\overline{E}, u, \overline{W}, x),$ and the corresponding witness by $\mathcal{W} \triangleq (E, r_E, W, r_W),$ where $\overline{E}$ (resp.\ $\overline{W}$) is a commitment to $E$ (resp.\ $W$) with opening randomness $r_E$ (resp.\ $r_W$).

The NIFS protocol comprises the following algorithms which we utilize in our system design:
\begin{itemize}
    \item $\mathsf{pp} \gets \mathsf{NIFS.ParaGen}(1^\lambda)$: Generates public parameters $\mathsf{pp}$.
    \item $(\mathsf{pk}, \mathsf{vk}) \gets \mathsf{NIFS.KeyGen}(\mathsf{pp}, (A,B,C))$: Generates keys for the constraint system $(A,B,C)$.
    \item $(\mathcal{I}, \mathcal{W}, \overline{T}) \gets \mathsf{NIFS.Process}(\mathsf{pk}, (\mathcal{I}_1,\mathcal{W}_1), (\mathcal{I}_2,\mathcal{W}_2))$: Folds two instance-witness pairs $(\mathcal{I}_1, \mathcal{W}_1),~(\mathcal{I}_2, \mathcal{W}_2)$ into a single pair $(\mathcal{I}, \mathcal{W})$ and outputs a cross-term commitment $\overline{T}$ (a commitment to $T$ under $\mathsf{Com}$).
    \item $\mathcal{I} \gets \mathsf{NIFS.Verify}(\mathsf{vk}, \mathcal{I}_1, \mathcal{I}_2, \overline{T})$: Deterministically computes the folded instance from public inputs and $\overline{T}$ to check consistency.
\end{itemize}
Both the underlying zkSNARK and Nova's folding scheme satisfy standard completeness, soundness, and zero-knowledge properties~\cite{nova}, ensuring that valid proofs can be generated if and only if a corresponding witness exists, without revealing the witness itself.
% ====== (END) Copied from the main paper Appendix app:nova ======

\section{Multi-Key Homomorphic Encryption}\label{supp:mkhe}

% ====== (BEGIN) Copied from the main paper Appendix app:mkhe ======
Multi-Key Fully Homomorphic Encryption (MKHE) allows computation over data encrypted under different keys, ensuring that intermediate values and final results remain encrypted. A key feature is collaborative decryption: no single party can access the plaintext. In the current ZK-AMS prototype, recovering the folded batch result requires aggregating partial decryption shares from all secret keys involved in the finalized computation; the system does not instantiate a separate tunable $t$-out-of-$N$ threshold-decryption layer.

The MKHE scheme consists of the following algorithms:
\begin{itemize}
    \item $\mathsf{pp} \leftarrow \mathsf{MKHE.ParamGen}(1^\lambda, 1^K)$: Parameter generation. Generates public parameters $\mathsf{pp}$ including any public evaluation material needed for homomorphic multiplication, given the security parameter $\lambda$ and the maximum number of parties $K$.

    \item $(\mathsf{pk}_i, \mathsf{sk}_i) \leftarrow \mathsf{MKHE.KeyGen}(\mathsf{pp})$: Key generation. Outputs a public key $\mathsf{pk}_i$ and a secret key $\mathsf{sk}_i$ for each party $i \in \{1, \dots, K\}$.

    \item $c_i \leftarrow \mathsf{MKHE.Encrypt}(\mathsf{pk}_i, \mu)$: Encryption. Encrypts a message $\mu$ under the public key $\mathsf{pk}_i$ for a party $i \in \{1, \dots, K\}$ to output a ciphertext vector $c_i$.

    \item $\hat{c}_S \leftarrow \mathsf{MKHE.Add}(c_1, c_2)$ (denoted as $c_1 \oplus c_2$): Homomorphic addition. Given two ciphertext vectors $c_1, c_2$, outputs an evaluated ciphertext $\hat{c}_S$ that decrypts to the sum of the underlying plaintexts.

    \item $\hat{c}_S \leftarrow \mathsf{MKHE.Multiply}(\mathbf{M}, c)$ (denoted as $\mathbf{M} \cdot c$): Homomorphic matrix multiplication. Given a plaintext matrix $\mathbf{M}$ and a ciphertext vector $c$, outputs an evaluated ciphertext $\hat{c}_S$ that decrypts to the product of $\mathbf{M}$ and the underlying plaintext vector.

    \item $\hat{c}_S \leftarrow \mathsf{MKHE.Multiply}(a, c)$ (denoted as $a \odot c$): Homomorphic scalar multiplication. Given a scalar $a$ and a ciphertext vector $c$, outputs an evaluated ciphertext $\hat{c}_S$ that decrypts to the product of $a$ and the underlying plaintext of $c$.

    \item $\hat{c}_S \leftarrow \mathsf{MKHE.HMul}(c_1, c_2)$ (denoted as $c_1 \otimes c_2$): Ciphertext-ciphertext (entry-wise) multiplication. Outputs an evaluated ciphertext vector encrypting the Hadamard product of the underlying plaintext vectors.

    \item $p_i \leftarrow \mathsf{MKHE.PartialDecrypt}(\mathsf{sk}_i, \hat{c}_S)$: Partial decryption. Outputs a partial decryption share $p_i$ using the secret key $\mathsf{sk}_i$ for the evaluated ciphertext $\hat{c}_S$, where $i$ is in the set of involved parties $S$.

    \item $\mu \leftarrow \mathsf{MKHE.Combine}(\{p_i\}_{i \in S}, \hat{c}_S)$: Decryption combination. Reconstructs the plaintext message $\mu$ by aggregating the partial shares $\{p_i\}_{i \in S}$ from all parties in the set $S$ associated with $\hat{c}_S$.
\end{itemize}

In ZK-AMS, we leverage a BGV-style multikey homomorphic encryption (MKHE) scheme \cite{LopezAlt2012MultikeyFHEMPC, BrakerskiGentryVaikuntanathan2012} to secure the witness folding process in Nova. Users (acting as parties in the BGV MKHE scheme) encrypt their private witnesses, enabling a distributed aggregation process to homomorphically fold them according to Nova constraints. This design eliminates the need for a centralized aggregation node with access to plaintext witnesses during folding; the PBS only reconstructs a randomized folded witness required to generate the final zkSNARK proof.

This all-party share-release design is sufficient for the current prototype, but it also explains two limitations highlighted in the main manuscript. First, confidentiality is not claimed against a coalition consisting of the PBS and all but one participant in a finalized batch, because after share release that remaining participant's folded plaintext contribution to the revealed accumulator may be isolated by subtracting the colluders' known contributions. This is weaker than direct recovery of the participant's entire original witness, but it still falls outside the confidentiality claim made in the paper. Second, malformed or withheld decryption shares can stall batch finalization; thus the present prototype relies on off-chain timeout and re-queue policies rather than publicly verifiable share validation or identifiable-abort mechanisms~\cite{StadlerPVSS1996,IshaiOstrovskyZikas2014IdentifiableAbort}. These limitations do not affect the correctness of already confirmed on-chain state, but they remain important deployment constraints of the current MKHE-backed batching layer.
% ====== (END) Copied from the main paper Appendix app:mkhe ======

\section{Linkable Ring Signatures}\label{supp:lrs}

% ====== (BEGIN) Copied from the main paper Appendix app:lrs ======
Linkable Ring Signatures (LRS) \cite{rivest2001leak, beullens2020calamari} allow a signer to sign a message on behalf of a group anonymously while enabling the detection of signatures produced by the same signer. An LRS scheme consists of four algorithms:
\begin{itemize}
\item $(\mathsf{pk, sk}) \leftarrow \mathsf{LRS.KeyGen}(1^k)$: Generates a public-private key pair given security parameter $k$.
\item $\sigma \leftarrow \mathsf{LRS.Sign}(1^k, 1^n, m, L, \mathsf{sk})$: Outputs a signature $\sigma$ for message $m$ using secret key $\mathsf{sk}$ and a ring of $n$ public keys $L$ (where the signer's $\mathsf{pk} \in L$).
\item $b \leftarrow \mathsf{LRS.Verify}(1^k, 1^n, m, L, \sigma)$: Outputs $b=1$ (accept) if $\sigma$ is a valid signature on $m$ with respect to ring $L$, otherwise $b=0$.
\item $b \leftarrow \mathsf{LRS.Link}(1^k, 1^n, m_1, m_2, \sigma_1, \sigma_2, L_1, L_2)$: Outputs $b=1$ if $\sigma_1$ and $\sigma_2$ were generated by the same signer (derived from the same $\mathsf{sk}$), otherwise $b=0$.
\end{itemize}
The scheme satisfies correctness, unforgeability, signer ambiguity, and linkability \cite{noether2016ring}. In ZK-AMS, we adopt the Multilayered Linkable Spontaneous Anonymous Group Signature (MLSAGS) \cite{noether2016ring}, which utilizes a unique key image $y_0$ to enforce linkability. The key image is computed as: $y_0 = \mathsf{sk} \cdot H_p(\mathsf{pk})$, where $H_p$ is a deterministic hash function mapping to an elliptic curve point. Since $y_0$ is mathematically bound to the signer's key pair $(\mathsf{sk}, \mathsf{pk})$, it serves as a unique tag to prevent double-signaling (e.g., preventing one user from binding multiple accounts) without revealing the user's identity.
% ====== (END) Copied from the main paper Appendix app:lrs ======

\section{Extended Security Proof Sketches}\label{supp:security}

This section complements the system-security analysis in Section~V of the main manuscript. Because ZK-AMS is a principled composition of existing primitives rather than a new primitive, the goal here is not to re-prove the adopted zkSNARK, Nova, MLSAGS, or MKHE constructions from first principles. Instead, we make the reduction logic behind the main manuscript's six propositions more explicit by documenting which primitive-level security notion or deployment assumption each system-level guarantee ultimately relies on.

\subsection{Assumption-to-guarantee map}

For quick reference, the guarantees G1--G6 in the main manuscript rest on the following assumption paths.
\begin{itemize}
    \item \textbf{G1 (Admission integrity):} Assumptions A1--A3; reduction target is issuer/holder signature unforgeability or zkSNARK soundness.
    \item \textbf{G2 (Duplicate-admission resistance):} Assumptions A1--A3; reduction target is collision resistance of $\mathsf{hash_{PHC}}$, zkSNARK soundness, or violation of issuer-side uniqueness in A1.
    \item \textbf{G3 (Credibly anonymous provisioning):} Assumptions A2--A3; reduction target is MLSAGS correctness, unforgeability, or anonymity.
    \item \textbf{G4 (Single-account enforcement):} Assumption A2; reduction target is MLSAGS linkability.
    \item \textbf{G5 (Witness confidentiality during batching):} Assumptions A2 and A5; reduction target is BGV-style MKHE semantic security together with the non-trivial coalition bound stated in A5.
    \item \textbf{G6 (Batch consistency and settlement integrity):} Assumptions A2--A3; reduction target is Nova/NIFS transcript consistency in the Fiat--Shamir random-oracle setting or zkSNARK soundness.
\end{itemize}

These mappings should be read as \emph{system-level composition claims}. They show that any successful attack against a claimed ZK-AMS property would imply either a break of a named primitive or a failure of an explicitly stated deployment assumption.

\subsection{Per-goal proof sketches}

\paragraph*{G1: Admission integrity.}
Suppose an adversary causes the Verifier Contract to accept a batch containing an admission for which it does not know a valid PHC and corresponding holder key. The on-chain verifier accepts only if the submitted zkSNARK proof verifies with respect to the committed admission relation. Hence either the adversary has produced a false accepting proof for a statement outside the relation, which contradicts zkSNARK soundness, or the accepted statement itself passes the embedded credential checks. In the latter case, the adversary must have fabricated a valid issuer-authenticated PHC or a valid holder-binding statement without possessing the proper signing material, which contradicts the unforgeability of the underlying issuer or holder signature component. Assumption A1 then rules out the degenerate case in which the credential is syntactically valid but no longer represents a legitimate personhood binding.

\paragraph*{G2: Duplicate-admission resistance.}
The contract records $\mathsf{hash_{PHC}}$ values as on-chain admission anchors and atomically rejects repeated anchors during settlement. Therefore, an accepted duplicate admission cannot arise from simply replaying the same public anchor after confirmation. A successful bypass must instead fall into one of three categories. First, two distinct PHCs may collide under $\mathsf{hash_{PHC}}$, contradicting collision resistance. Second, the proof may falsely certify that the batch relation is satisfied even though the duplicate check should fail, contradicting zkSNARK soundness. Third, the issuer may have violated the uniqueness trust root in A1 by issuing multiple simultaneously valid PHCs for the same real-world person. Thus the duplicate-admission guarantee is exactly as strong as those underlying primitives and trust assumptions.

\paragraph*{G3: Credibly anonymous provisioning.}
The Soul Registry accepts a provisioning request only if the submitted MLSAGS signature verifies against a ring of admitted seed public keys. Consequently, if an adversary causes acceptance without any admitted signer, it must have forged a valid MLSAGS signature or broken its correctness guarantees. Conversely, conditioned on the existence of some admitted signer, the identity of that signer is hidden among the ring members up to the anonymity bound of MLSAGS. Any non-negligible advantage in identifying which admitted seed key actually signed can therefore be transferred into an adversary against MLSAGS anonymity. This is why ZK-AMS can simultaneously claim public verifiability of authorization and signer anonymity.

\paragraph*{G4: Single-account enforcement.}
In ZK-AMS, successful provisioning writes the MLSAGS key image $y_0$ to the Soul Registry, and future provisioning attempts that reuse the same image are rejected. For a fixed seed signing secret key, MLSAGS linkability requires all valid signatures under that key to induce the same key image. Therefore, if an adversary manages to derive two distinct accepted Soul-Account bindings from the same admitted seed identity without triggering key-image reuse, the adversary would violate MLSAGS linkability. The system-level one-to-one provisioning guarantee is thus inherited directly from the primitive-level uniqueness of the MLSAGS key image.

\paragraph*{G5: Witness confidentiality during batching.}
The confidentiality argument has two layers. Before decryption-share release, the PBS and any colluding users see only MKHE ciphertexts, transcript metadata, and public commitments. A standard hybrid argument can then replace the encrypted witnesses of the non-colluding users with encryptions of dummy values, one by one, without changing the adversary's view except with negligible advantage; otherwise one obtains a distinguisher against BGV-style MKHE semantic security. After share release, the coalition also learns the final folded plaintext accumulator. At this point the protection no longer comes from ciphertext indistinguishability alone, but from the coalition bound in A5: because at least two participants remain outside the PBS coalition, subtracting the colluders' known contributions reveals only a public linear combination of at least two unknown witness-dependent terms, not an isolated term attributable to any one remaining participant. The current prototype therefore does \emph{not} claim confidentiality against a coalition consisting of the PBS and all but one participant, because that excluded case may permit isolation of the last participant's folded plaintext contribution after decryption-share release. This is weaker than direct recovery of the participant's entire original witness, but it is already outside the confidentiality claim proved here.

\paragraph*{G6: Batch consistency and settlement integrity.}
Each folding step derives its Fiat--Shamir challenge from the public transcript. Therefore, if an adversary tampers with the commitments, ordering, or batch metadata while attempting to preserve the same folded state, it must either find a transcript inconsistency that survives the Nova/NIFS verification logic or create a false accepting final proof. More concretely, modifying any transcript-dependent element changes the challenge used to derive the next folded instance, unless the adversary breaks the random-oracle-based consistency assumptions behind the Fiat--Shamirized Nova/NIFS transcript. If the adversary instead submits an inconsistent accumulated state together with a proof that still passes contract verification, then zkSNARK soundness is violated. Hence transcript tampering that evades settlement checks reduces to one of these two primitive failures.

Taken together, these extended sketches explain the intended meaning of the six propositions in the main manuscript: they are not standalone primitive-security theorems, but explicit composition claims showing how ZK-AMS inherits its guarantees from the cited building blocks plus the stated deployment assumptions.